\documentclass[journal,onecolumn]{IEEEtran}
\usepackage{caption}
\usepackage{amsmath}
\usepackage{amssymb}
\usepackage{amsfonts}
\usepackage{booktabs}
\usepackage{mathrsfs}
\usepackage{graphicx}
\usepackage{subfigure}
\usepackage{multirow}
\ifCLASSOPTIONcompsoc
\usepackage[caption=false,font=normalsize,labelfon
t=sf,textfont=sf]{subfig}
\else
\usepackage[caption=false,font=footnotesize]{subfig}
\fi

\newtheorem{theorem}{Theorem}

\newtheorem{corollary}{Corollary}

\newtheorem{proposition}{Proposition}
\newtheorem{example}{Example}
\newtheorem{property}{Property}

\begin{document}

\title{Maximum Likelihood Estimation Based Complex-Valued Robust Chinese Remainder Theorem and Its Fast Algorithm}

\author{Xiaoping Li,  \IEEEmembership{Member, IEEE,}
 Shiyang Sun, Qunying Liao, Xiang-Gen Xia, \IEEEmembership{Fellow, IEEE}

\thanks{The work of Xiaoping Li was supported in part by the National Natural Science Foundation of China under Grant 62131005. The work of Qunying Liao was supported in part by the National Natural Science Foundation of China under Grant 12471494. The work of Xiang-Gen Xia was supported in part by the National Science Foundation (NSF) under Grant CCF-2246917.

Xiaoping Li and Shiyang Sun are with the School of Mathematical Science, University of Electronic Science and Technology of China, Chengdu 611731, China. (e-mail: {\tt lixiaoping.math@uestc.edu.cn}; {\tt 202421110224@std.uestc.edu.cn}).

Qunying Liao is with the Institute of Mathematics Science, Sichuan Normal University, Chengdu 610068, China. (e-mail: {\tt qunyingliao@sicnu.edu.cn}).

Xiang-Gen Xia is with the Department of Electrical and Computer Engineering, University of Delaware, Newark, DE 19716, USA. (e-mail: {\tt xxia@ee.udel.edu}).}}

\markboth{Journal of \LaTeX\ Class Files,~Vol.~14, No.~8, August~2015}%
{Shell \MakeLowercase{\textit{et al.}}: Bare Demo of IEEEtran.cls for IEEE Journals}

\maketitle
\begin{abstract}
In this paper, we investigate complex-valued Chinese remainder theorem (C-CRT) with erroneous remainders, 
where the  moduli are Gaussian integers and the errors follow wrapped complex Gaussian distributions. 
Based on the existing real-valued CRT utilizing maximum likelihood estimation (MLE), we propose a fast MLE-based C-CRT (MLE C-CRT). 
The proposed algorithm requires only $2L$ searches to obtain the optimal estimate of the common remainder, 
where $L$ is the number of moduli. 
Once the common remainder is estimated, the complex number can be determined using the C-CRT. Furthermore, 
we obtain a necessary and sufficient condition for the fast MLE C-CRT to achieve robust estimation. 
Finally, we apply the proposed algorithm to a multi-channel self-reset analog-to-digital converter (ADC) system with Gaussian integers as moduli, 
which enables the recovery of high dynamic range complex-valued bandlimited signals at the Nyquist sampling rate. The results demonstrate that the proposed algorithm  outperforms the existing methods.
\end{abstract}


\ifCLASSOPTIONpeerreview
\begin{center} \bfseries EDICS Category: 3-BBND \end{center}
\fi

\IEEEpeerreviewmaketitle

\section{Introduction}
\IEEEPARstart{T}{he} Chinese remainder theorem (CRT) is a fundamental theorem in ring theory, widely applied in computer science, coding theory, and digital signal processing \cite{Krishna_1994, Ding_1999}. However, the CRT is not robust as even a small error in any remainder may lead to a large error in the reconstruction. 
To overcome this shortcoming, a robust CRT has been studied in \cite{Xia_2007, Li_2007, Li_2008, Li_2009, Wang_2010, Yang_2014, Xiao_2014, Wang_2015} by utilizing remainder redundancy. The existing literature mainly considers two types of remainder redundancy: 1) the remaining factors of the moduli after being divided by their greatest common divisor (gcd) are pairwise coprime; and 2) the remaining factors of the moduli after being divided by their gcd are not pairwise coprime. For the first type of remainder redundancy, any two moduli have the same gcd. Furthermore, all the remainders modulo the gcd are identical and are referred to as the common remainder \cite{Ore_1952}. In \cite{Xia_2007}, a searching-based method is proposed to address this redundancy. In \cite{Wang_2010}, a closed-form CRT is introduced, assuming identical remainder error variances, thereby eliminating the need for search steps through a direct closed-form reconstruction process. In \cite{Wang_2015}, a maximum likelihood estimation (MLE)-based algorithm is proposed, which optimally estimates the common remainder and the noises may have different variances. For the second type of remainder redundancy, there are at least two distinct groups of moduli, each having a different gcd \cite{Yang_2014, Xiao_2014}. In \cite{Xiao_2014}, a multi-stage robust CRT method is proposed that enhances the robustness, with a potentially improved performance compared to the first type when the moduli are appropriately grouped. The robust CRT has numerous applications, such as in multi-channel SAR and InSAR systems \cite{Xu_2018, Huang_2024}. It has also been generalized for vectors \cite{Xiao_2020, Xiao_2024}, for multiple integers \cite{Xia_1997, Xia_1999, Wang_GCRT_2015, Li_2016, Li_2019}, and for polynomials \cite{Xiao_2015_p, Xiao_2018_p}.

In this paper, we propose robust complex-valued CRT (C-CRT) with Gaussian integers as moduli to robustly determine a complex number from its remainders modulo several Gaussian integers. 
It can be thought of as a  generalization of the robust CRT for real numbers in \cite{Wang_2010, Wang_2015}.
Note that the robust CRT for real numbers  is able to have the reconstruction error level the same as that of the remainder errors (or noises) that are typically measured using the circular distance based on the modulo operation \cite{Wang_2015}. However, the circular distance between two complex numbers involves both scaling and rotation in the complex plane. Hence, an error must be computed by considering both the real and imaginary parts simultaneously. Moreover, when the moduli are complex numbers, the reconstruction process becomes more complicated. When the complex moduli are pairwise conjugate each other, the reconstruction of a complex number can be solved by the two-stage robust CRT for real numbers \cite{Gong_2021}. 
To the best of our knowledge, there is no robust C-CRT that addresses the reconstruction of a complex number from its erroneous remainders directly. 
In this paper, we propose a robust C-CRT with Gaussian integers as moduli, where the product and the gcd of the moduli are real-valued integers (or simply called integers). 
Additionally, the moduli are pairwise coprime after divided by their gcd. Motivated by the work of \cite{Wang_2015}, we propose a fast MLE-based algorithm for the robust C-CRT for complex numbers, which is more challenging compared to that for real numbers, particularly in the MLE model and the analysis of robust estimation. 

\begin{figure}[h]
\centering
\captionsetup{skip=-10pt}
\setlength{\belowcaptionskip}{-0.2cm}
\captionsetup{font={footnotesize}}
\captionsetup{labelsep=period}
\includegraphics[width=0.5\textwidth]{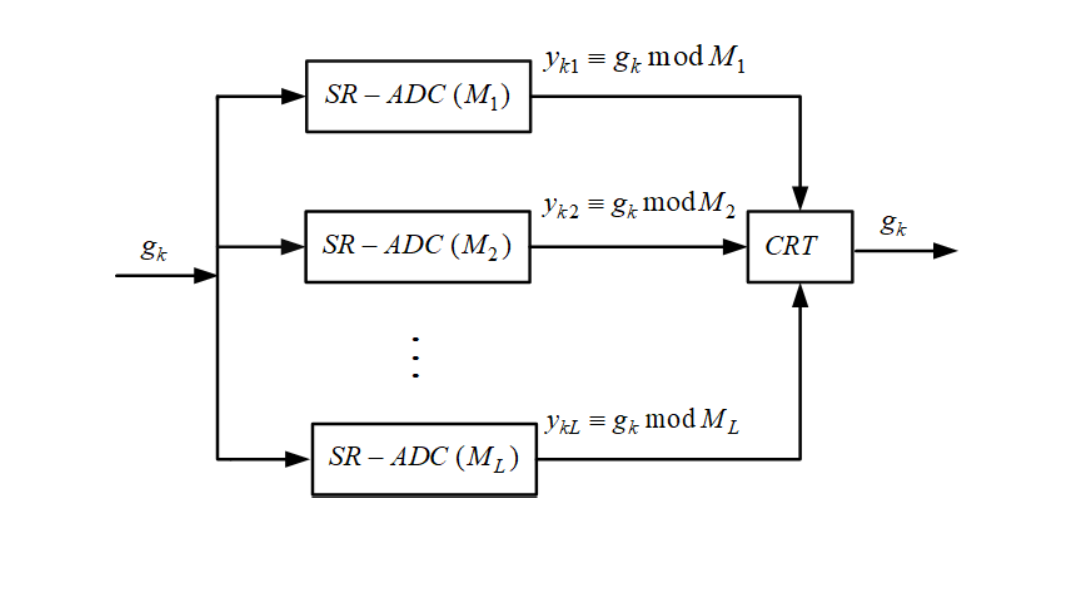}
\caption{Multi-channel SR-ADCs \cite{Gan_2020}.}\label{fig:r_SR_ADC}
\end{figure}

As we shall see later, although C-CRT can be formulated as a special case of 2D-CRT \cite{Xiao_2020, Lin_1994} and robust MD-CRT has been studied in a general setting \cite{Xiao_2024}, and both MD-CRT and robust MD-CRT can be generalized to real-valued vectors, the study in this paper is much different in the following sense. The key difference with the studies for robust MD-CRT in \cite{Xiao_2020} and \cite{Xiao_2024} is that in this paper, we take a probabilistic approach and treat remainder errors as random variables, while the studies in \cite{Xiao_2020} and \cite{Xiao_2024} are deterministic only. In this paper we propose the MLE C-CRT and robust MLE C-CRT, when the remainder errors follow wrapped complex Gaussian distributions that are the most common distributions for remainder errors. The results we obtain in this paper also provide a detailed (special and interesting) robust 2D-CRT.

The proposed fast MLE C-CRT algorithm can be applied in the multi-channel self-reset analog-to-digital converter (SR-ADC) system proposed in \cite{Gan_2020, Gong_2021, Yan_2024, Yan_2025}, as shown in Fig. \ref{fig:r_SR_ADC}. Compared to the single-channel ADC system, it offers a higher dynamic range and is capable of reconstructing bandlimited signals with sampling at the Nyquist rate. Unlike the traditional ADC-based continuous-time signal recovery, this system first reconstructs the sampled values from their multiple modulo samples before recovering a continuous-time signal. In \cite{Gong_2021}, the authors have proposed a complex-valued modulus multi-channel ADC architecture based on Gaussian integers, which offers a higher dynamic range compared to \cite{Gan_2020}. To recover sampled values of a complex-valued bandlimited signal, the moduli are classified into two types: Gaussian integers and positive integers. A sampled value is then reconstructed using the two-stage robust CRT \cite{Xiao_2014} for a real number. In the first stage, the remainders of the Gaussian integer moduli are used to recover the partial signal value through the closed-form robust CRT \cite{Wang_2010}. In the second stage, the remainders of the positive integer moduli are utilized to recover the complete signal value based on the first stage. It is demonstrated that the reconstruction is robust if the error conditions for both stages are satisfied.

Our contributions are fourfold. First, we propose an efficient algorithm for the C-CRT to determine the MLE from erroneous remainders with wrapped complex Gaussian noises. Second, we provide the optimal estimate of the common remainder from complex erroneous remainders. The total number of the optimal common remainder candidates is $2L$ compared to $L$ in the real-valued case presented in \cite{Wang_2015}, where $L$ is the number of moduli. Third, we derive a necessary and sufficient condition for the MLE-based C-CRT (MLE C-CRT) to be robust. 
Forth, we apply our proposed robust C-CRT to multi-channel SR-ADCs with improved performance.

The remainder of this paper is organized as follows. In Section \ref{Background}, we introduce modulo operations for complex numbers and the C-CRT in the absence of errors. Moreover, we introduce the application background of C-CRT in the modulus sampler. In Section \ref{sec_MLE}, we present a fast MLE C-CRT algorithm to
estimate a complex number from erroneous complex remainders. In Section \ref{sec4}, we provide a necessary and sufficient condition for the MLE C-CRT to be robust. In Section \ref{sec5}, we present simulation results to verify the performance of the proposed algorithm and demonstrate its application to ADCs.

\textit{Notations}: The set of integers is denoted as $\mathbb{Z}$, and the sets of 2 dimensional (2D) real vectors and integer vectors are denoted as $\mathbb R^2$ and $\mathbb Z^2$, respectively. To clearly distinguish between complex numbers, real numbers, and matrices, this paper uses $\mathsf{N}$, $\mathsf{\Gamma}$, $\mathsf{r}$, etc., to represent complex numbers; $N$, $\Gamma$, $r$, etc., for real numbers; and $\mathbf N $, $\mathbf M $, $\mathbf k$, etc., for matrices.
For a complex number $\mathsf z = z_1 + z_2 \mathrm{i}$, where $\mathrm{i}$ represents the imaginary unit, i.e., $\mathrm{i} = \sqrt{-1}$, $\mathrm{Re}(\mathsf z)$ denotes the real part $z_1$, and $\mathrm{Im}(\mathsf z)$ denotes the imaginary part $z_2$. $\left\lfloor r \right\rfloor$ denotes the flooring operation of real number $r$, i.e., the greatest integer less than or equal to $r$, and $\left\lfloor\mathsf z\right\rfloor$ denotes the flooring operation of $\mathsf z$, i.e., $\left\lfloor\mathsf z\right\rfloor = \left\lfloor \mathrm{Re}(\mathsf z)\right\rfloor + \left\lfloor \mathrm{Im}(\mathsf z)\right\rfloor\mathrm i$. The set $\mathbb{Z}[\mathrm{i}] = \{ z_1 + z_2 \mathrm{i}: z_1, z_2\in \mathbb{Z} \}$ is the ring of Gaussian integers.

\section{C-CRT and Problem Description}\label{Background}
In this section, we first introduce some concepts for C-CRT, including the Euclidean division, the system of complex congruences, and the Euclidean algorithm. We then discuss the application background and the challenges of C-CRT in modulo samplers.

\subsection{Basic Concepts for C-CRT}
First, we introduce the Euclidean division for complex numbers. Let $\mathsf N$ be a complex number, and let $\mathsf M$ be a nonzero Gaussian integer. Then, there exist unique $\mathsf r\in \mathcal F_{\mathsf M}$ and $\mathsf q\in \mathbb Z[\mathrm i]$ satisfying
\begin{eqnarray}\label{C_div}
\mathsf N = \mathsf M \mathsf q + \mathsf r,
\end{eqnarray}
where $\mathcal F_{\mathsf M}$ is the complex remainder set satisfying
\begin{eqnarray}\label{FM}
\mathcal F_{\mathsf M} = \left\{\mathsf M(a + b\mathrm i): 0 \leq a, b < 1 \right\},
\end{eqnarray}
and $\mathsf r$ is called the remainder of $\mathsf N$ modulo $\mathsf M$.

For $\mathcal F_{\mathsf M}$ described above, let $\mathsf M = \rho e^{\mathrm i\theta}$, where $\rho$ and $\theta$ represent the modulus and angle of $\mathsf M$, respectively. Then, we have two properties below, which are proven in Appendix $A$.

\begin{property}\label{prop1_F_M}
If $\mathsf z\in \mathcal F_{\mathsf M}$, then $\mathsf ze^{-\mathrm i\theta}\in \mathcal F_{\mathsf \rho }$.
\end{property}

\begin{property}\label{prop2_F_M}
$\mathcal F_{\mathsf M}$ is a square and its area is $\rho^2$.
\end{property}

If we consider the real and imaginary parts of both sides of (\ref{C_div}) separately, then we have
\begin{eqnarray}\label{n1n2_matirx}
\begin{pmatrix} n_1 \\   n_2  \end{pmatrix} = \begin{pmatrix} m_1 & -m_2 \\ m_2 & m_1\end{pmatrix}\begin{pmatrix} q_1 \\  q_2 \end{pmatrix} + \begin{pmatrix} r_1 \\  r_2  \end{pmatrix},
\end{eqnarray}
where $\mathsf N =  n_1 +  n_2\mathrm i$, $\mathsf M =  m_1 +  m_2\mathrm i$, $\mathsf q =  q_1 + q_2\mathrm i$, and $\mathsf r = r_1 +  r_2\mathrm i$. One can see that (\ref{n1n2_matirx}) is the 2D modulo problem studied in MD-CRT in \cite{Xiao_2020, Xiao_2024, Lin_1994} with integer matrix moduli of the form
\begin{eqnarray}\label{MathbfM}
\mathbf M = \begin{pmatrix} m_1 & - m_2 \\ m_2 & m_1\end{pmatrix}.
\end{eqnarray}
The set $\mathcal F_{\mathsf M}$ in (\ref{FM}) is equivalent to the following set of vector remainders modulo $\mathbf M$, which is known as the fundamental parallelepiped (FPD) of $\mathbf M$ \cite{Vaidyanathan_1993}:
$$\text{FPD}(\mathbf M) = \{\mathbf k: \mathbf k = \mathbf M \mathbf x, \mathbf x \in [0, 1)^2\}.$$
In the following, for notational convenience, for any complex number $\mathsf N = n_1 + n_2\mathrm i $, we use its equivalent forms
$n_1 +  n_2\mathrm i$, $\begin{pmatrix} n_1 \\  n_2  \end{pmatrix}$, and $\begin{pmatrix}n_1 & - n_2 \\  n_2 &  n_1\end{pmatrix}$
interchangeably, whenever and wherever they apply. For example, when complex number $\mathsf M = m_1 + m_2\mathrm i$ is used as a modulus number, it is either $m_1 + m_2\mathrm i$ or $\begin{pmatrix} m_1 & -m_2 \\ m_2 & m_1\end{pmatrix}$. In what follows, we consider general complex numbers or vectors in $\mathcal F_{\mathsf M}$ or $\text{FPD}(\mathbf M)$, not just Gaussian integers or lattice points, while the moduli are Gaussian integers only.

By (\ref{C_div}) and (\ref{FM}), $\mathsf r$ can be rewritten as
\begin{equation}\label{mathsfrN}\textstyle
\mathsf r = \mathsf N - \mathsf M\left\lfloor \frac{\mathsf N}{\mathsf M} \right\rfloor.
\end{equation}
For convenience, we denote
$\left\langle \mathsf N\right\rangle_\mathsf M = \mathsf r = \mathsf N \mod \mathsf M$.

\begin{figure}[h]
\centering
\setlength{\belowcaptionskip}{-0.2cm}
\captionsetup{font={footnotesize}}
\captionsetup{labelsep=period}
\setlength{\subfigcapskip}{-15pt}
\subfigure[$\mathsf M=3+4\mathrm i$.]{\includegraphics[width=0.25\textwidth]{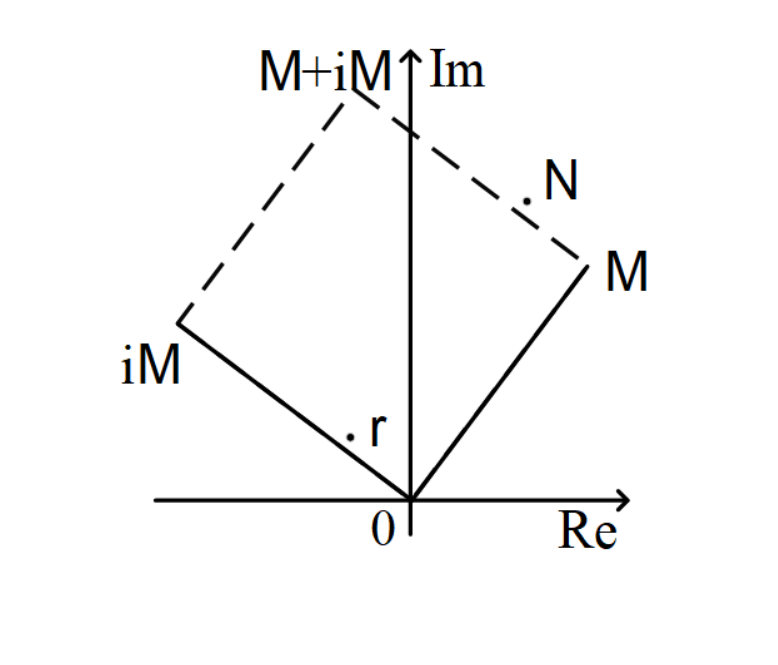}}
\subfigure[$\mathsf M=4$.]{\includegraphics[width=0.25\textwidth]{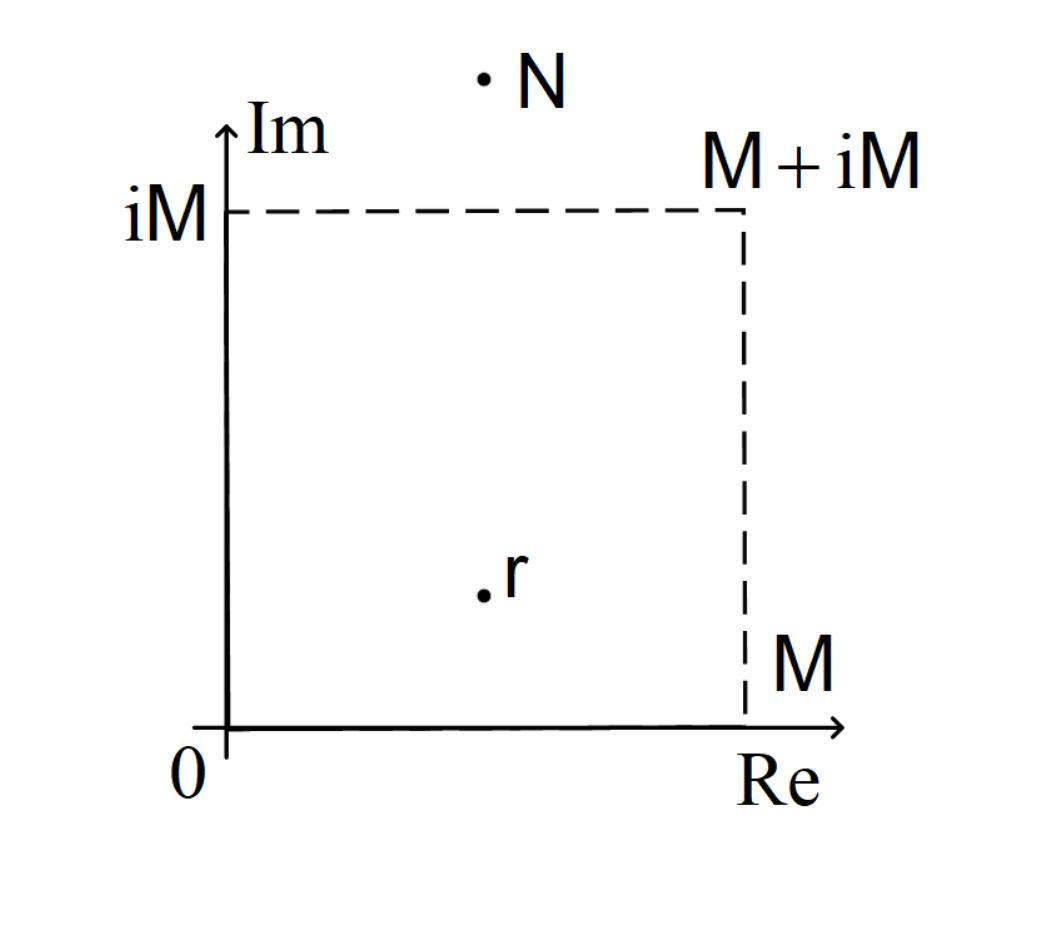}}
\caption{Illustration of $\mathcal F_{\mathsf M}$.}\label{F_M_1_2}
\end{figure}

Fig. \ref{F_M_1_2} gives an illustration of the complex remainder set $\mathcal F_{\mathsf M}$ when $\mathsf N = 2+5\mathrm{i}$. By (\ref{mathsfrN}), we have $\left\langle \mathsf N\right\rangle_{3+4\mathrm{i}} = -1+\mathrm{i}$ and $\left\langle \mathsf N\right\rangle_{4} = 2+\mathrm{i}$. Clearly, if $\mathsf M$ is a real number, then
$$ \langle \mathsf N \rangle_{\mathsf M} = \langle \mathrm{Re}(\mathsf N)
\rangle_{\mathsf M} + \mathrm i \langle \mathrm{Im}(\mathsf N)\rangle_{\mathsf M}.$$
In this case, the modulo operation is performed separately on the real and imaginary parts of $\mathsf N$.

The C-CRT replaces moduli and remainders of the real-valued CRT with complex values. Now, we introduce the C-CRT when moduli are Gaussian integers. Unlike the CRT in rings, where the remainders must belong to an integral domain (see Theorem 17, Section 7.6 in \cite{Dummit_2004}), this C-CRT allows the remainders to be any complex numbers, not just limited to Gaussian integers similar to the CRT for real numbers studied in \cite{Wang_2010, Wang_2015}. The problem is as follows. Let $\mathsf\Gamma_i$, $i = 1, 2, \ldots, L$, be $L$ Gaussian integers as moduli, $\mathsf N$ be a complex number, and $\mathsf r_i = \mathsf N \mod \mathsf\Gamma_i$, $i = 1, 2, \ldots, L$, be its remainders, where $\mathsf N$ and $\mathsf r_i$, $i = 1, 2, \ldots, L$, may not necessarily be Gaussian integers as explained above. Thus, we have the following system of congruences: 
\begin{equation}\label{eq_C_CRT}
\mathsf N = \mathsf k_i \mathsf\Gamma_i + \mathsf r_i, \ i = 1, 2, \ldots, L,
\end{equation}
where $\mathsf k_i$, $i = 1, 2, \ldots, L$, are unknown Gaussian integers called folding Gaussian integers. The problem is to determine $\mathsf N$ from its remainders $\mathsf r_i$, $i = 1, 2, \ldots, L$. This problem occurs in multi-channel SR-ADC for complex-valued bandlimited signals \cite{Gong_2021}.

We next present C-CRT. For two Gaussian integers $\mathsf\Gamma_i$ and $\mathsf\Gamma_j$, we say that $\mathsf\Gamma_i$ and $\mathsf\Gamma_j$ are coprime if their common divisors are all $\pm 1$ and $\pm\mathrm i$. This coprimality is consistent with that for 2D integer matrices, i.e., if $\mathbf\Gamma_i$ and $\mathbf\Gamma_j$ are treated as two 2D integer matrices, then two Gaussian integers $\mathsf\Gamma_i$ and $\mathsf\Gamma_j$ are coprime if and only if their corresponding two 2D integer matrices $\mathbf\Gamma_i$ and $\mathbf\Gamma_j$ are coprime \cite{Li_2019_def}.

\begin{theorem} [C-CRT] \label{Th_C_CRT}
Let $\mathsf\Gamma_1, \mathsf\Gamma_2, \ldots, \mathsf\Gamma_L$ be pairwise coprime Gaussian integers with $|\mathsf\Gamma_i| \ge \sqrt{2}$ for $i = 1, 2, \ldots, L$. Then, for a complex number $\mathsf N \in \mathcal F_{\mathsf\Gamma}$, the system of congruences \eqref{eq_C_CRT} has a unique solution
\begin{equation}\label{eq_N_mod_Gamma}
\mathsf N = \left\langle \mathsf r_1 - \left\lfloor \mathsf r_1 \right\rfloor + \sum_{i=1}^L\bar{\gamma}_i\mathsf\gamma_i \left\lfloor \mathsf r_i \right\rfloor \right\rangle_\mathsf\Gamma,
\end{equation}
where $\mathsf\Gamma = \prod_{i=1}^L\mathsf\Gamma_i$, $\mathsf\gamma_i = \frac{\mathsf\Gamma}{\mathsf\Gamma_i}$, and $\bar{\mathsf\gamma}_i$ is the modular multiplicative inverse of $\mathsf\gamma_i$ modulo $\mathsf\Gamma_i$, i.e., there exists a Gaussian integer $\bar{\mathsf\Gamma}_i$ such that
\begin{equation}\label{eq_Bezout_1}
\mathsf\gamma_i\bar{\mathsf\gamma}_i + \mathsf \Gamma_i \bar{\mathsf\Gamma}_i = 1.
\end{equation}
\end{theorem}

\begin{proof}
Let $\mathsf N' = \mathsf N - \left\lfloor \mathsf N \right\rfloor$ and
$\mathsf r_i' = \mathsf r_i - \left\lfloor \mathsf r_i \right\rfloor$  for $i = 1, 2, \ldots, L$. By \eqref{eq_C_CRT}, we have
\begin{equation}\label{eq_r_0}
\mathsf N' - \mathsf r_i' 
= \mathsf k_i\mathsf\Gamma_i + \left\lfloor \mathsf r_i \right\rfloor - \left\lfloor \mathsf N \right\rfloor 
\in \mathbb Z[\mathrm i],\ i = 1, 2, \ldots, L.
\end{equation}
Since $\mathsf N'\in \mathcal F_1$ and $ \mathsf r_i'\in \mathcal F_1$. We have 
$\mathrm{Re}(\mathsf N' - \mathsf r_i' ) \in (-1, 1)$ 
and 
$\mathrm{Im}(\mathsf N' - \mathsf r_i' ) \in (-1, 1)$.
Then, we obtain from \eqref{eq_r_0} that $\mathsf N' = \mathsf r_i'$ holds for each $i$. Hence, $$\mathsf r_1' = \mathsf r_2' = \cdots = \mathsf r_L'.$$
Consequently, \eqref{eq_C_CRT} can be rewritten as
$$\mathsf N - \mathsf r_1' = \mathsf k_i\mathsf\Gamma_i + \left\lfloor \mathsf r_i \right\rfloor ,\ i = 1, 2, \ldots, L.$$
Since $\mathsf\Gamma_1, \mathsf\Gamma_2, \ldots, \mathsf\Gamma_L$ are pairwise coprime, according to the CRT over rings, 
$$\mathsf N - \mathsf r_1' \equiv \sum_{i=1}^L\bar{\mathsf\gamma}_i\mathsf\gamma_i \left\lfloor \mathsf r_i \right\rfloor \mod \mathsf\Gamma.$$
That is, there exists a Gaussian integer $\mathsf k$ such that
$$\mathsf N - \mathsf r_1' - \sum_{i=1}^L\bar{\mathsf\gamma}_i\mathsf\gamma_i \left\lfloor \mathsf r_i \right\rfloor = \mathsf k \mathsf\Gamma.$$
Therefore, \eqref{eq_N_mod_Gamma} is a solution of \eqref{eq_C_CRT} in $\mathcal F_{\mathsf\Gamma}$.

Next, we prove the uniqueness of the solution in $\mathcal F_{\mathsf\Gamma}$. Let $\mathsf N'\in\mathcal F_\mathsf\Gamma$ be another solution. By (\ref{eq_C_CRT}), we have
$\mathsf{N'} \equiv \mathsf N\mod\mathsf\Gamma_i$. Hence,
$\mathsf\Gamma_i$ divides $\mathsf N' - \mathsf N$.
Since $\mathsf\Gamma_1, \mathsf\Gamma_2, \ldots, \mathsf\Gamma_L$ are pairwise coprime, we obtain that $\mathsf\Gamma$ divides $\mathsf N'-\mathsf N$.
Thus, there exists a Gaussian integer $\mathsf k=k_1+k_2\mathrm i$ such that
$\mathsf N' = \mathsf N + \mathsf k\mathsf\Gamma$. Since $\mathsf N \in\mathcal F_{\mathsf\Gamma}$, there exist $n_1,n_2\in [0,1)$ such that
$\mathsf N = \mathsf\Gamma (n_1+n_2\mathrm i)$. Consequently, $\mathsf N' = \mathsf\Gamma(k_1+n_1+(k_2+n_2)\mathrm i)$. It follows from $\mathsf N'\in\mathcal F_\mathsf\Gamma$ that $k_1+n_1, k_2+n_2\in[0, 1)$.
Since $k_1, k_2\in\mathbb Z$, we have $k_1=k_2=0$. Therefore, $\mathsf N'=\mathsf N$.
\end{proof}

\textbf{Remark 1}: Different from the CRT in rings, Theorem \ref{Th_C_CRT} gives a reconstruction method for a complex number in $\mathcal F_\mathsf\Gamma$ (not just a Gaussian integer) from its remainders. In fact, the CRT in the ring of Gaussian integers can be easily generalized from that in the ring of integers. The C-CRT described here is for any complex number, which comes from the applications where the unknown $\mathsf N$ and its remainders are complex numbers.

Note that the general 2D-CRT for 2D vectors studied in \cite{Xiao_2020} and \cite{Xiao_2024} may not have the concise form in (\ref{eq_N_mod_Gamma}) similar to the conventional CRT for real integers. 
Another advantage of the above C-CRT over the general 2D-CRT is that it may be more convenient to find a set of pairwise coprime Gaussian integers \cite{LHLi_2019} than that for 2D integer matrices. Furthermore, a necessary and sufficient condition was obtained in \cite{PP_2011} for 2D integer matrices of the form (\ref{MathbfM}) called skew-circulant matrices in \cite{PP_2011} to be coprime.

As explained in Introduction, the key for the robust CRT for real numbers in the literature is to have some redundancies in the remainders. 
One of such redundancies is to have a non-unit  gcd among all the moduli. 
We next consider the moduli of the forms $M \mathsf\Gamma_i$ where $M$ is the gcd of all the moduli. Thus, we have the following system of congruences
\begin{equation}\label{mathsfNmod}
\mathsf{N} = \mathsf k_i M\mathsf\Gamma_i + \mathsf r_i, \ i = 1, 2, \ldots, L,
\end{equation}
where $\Gamma = \prod_{i=1}^L\mathsf\Gamma_i$ and $M$ are both assumed positive integers, $\mathsf k_i$, $i = 1, 2, \ldots, L$, are the unknown folding Gaussian integers. In other words, both the gcd and the least common multiple (lcm) of the moduli are assumed integers and in this case $\mathcal F_{M\Gamma}$ is a square of sides on the real and imaginary axes. 
In the following we generalize Theorem \ref{Th_C_CRT} to solve (\ref{mathsfNmod}). 
By (\ref{mathsfNmod}), we have
$$\mathsf{r}_i \equiv \mathsf{N}\mod M.$$ 
Then,
$$\mathsf{r}_i\equiv\left\langle \mathsf{N}\right\rangle_M \mod M.$$
That is, all remainders modulo $M$ have the same value called the common remainder $\mathsf r^c$, which is in $\mathcal F_M$ and can be determined by
\begin{equation}\label{eq_rc}
\mathsf r^c \equiv \mathsf r_i \mod M, \ i=1, 2, \ldots, L.
\end{equation}
It follows that $\mathsf r_i - \mathsf r^c \in M\mathbb Z[\mathrm i]$. Let
\begin{equation}\label{eq_q_iNN0}
\mathsf q_i = \frac{\mathsf r_i - \mathsf r^c}{M} \ \text{and} \
\mathsf N_0 = \frac{\mathsf N - \mathsf r^c}{M}.
\end{equation}
Then, we have $\mathsf q_i \in \mathbb Z[\mathrm i]$ and
\begin{equation}\label{mathsfN0}
\mathsf N_0 \equiv \mathsf q_i \mod \mathsf{\Gamma}_i, \ i=1, 2, \ldots, L.
\end{equation}
If $\mathsf N\in \mathcal F_{M\Gamma}$, we have
$0\leq\left\lfloor\mathrm{Re}(\mathsf N) \right\rfloor, \left\lfloor\mathrm{Im}(\mathsf N) \right\rfloor < M \Gamma$.
This leads to
$0\leq\left\lfloor\frac{\mathrm{Re}(\mathsf N)}{M} \right\rfloor, \left\lfloor\frac{\mathrm{Im}(\mathsf N)}{M} \right\rfloor < \Gamma$.
Since $$ \mathsf N - \mathsf r^c = \left\lfloor \frac{\mathsf N}{M} \right\rfloor M = \left\lfloor \frac{\mathrm{Re}(\mathsf N)}{M} \right\rfloor M + \mathrm i\left\lfloor \frac{\mathrm{Im}(\mathsf N)}{M} \right\rfloor M,$$ we have $\mathsf N - \mathsf r^c \in \mathcal F_{M\Gamma}$, and consequently, $\mathsf N_0 = \frac{\mathsf N - \mathsf r^c}{M}\in \mathcal F_\Gamma$.
By (\ref{mathsfN0}) and Theorem \ref{Th_C_CRT}, we have
\begin{equation}\label{N_0_mod}
\mathsf N_0 = \left\langle\sum_{i=1}^L\bar{\mathsf\gamma}_i\mathsf\gamma_i\mathsf q_i\right\rangle_\Gamma.
\end{equation}
It follows from (\ref{eq_q_iNN0}) that
\begin{equation}\label{eq_N_r_c}
\mathsf N = M \mathsf N_0 + \mathsf r^c.
\end{equation}

\textbf{Remark 2}:
As mentioned earlier, the C-CRT can be viewed as a special 2D-CRT.
If the moduli in the 2D-CRT can be simultaneously diagonalized by integer matrices
into diagonal integer matrices, the 2D-CRT reduces to two individual 1D-CRTs.
The following example demonstrates that the C-CRT does not generally reduce to two individual 1D-CRTs. Let $\mathsf\Gamma_1=3+4\mathrm i$ and $\mathsf\Gamma_2=3-4\mathrm i$. Then the matrix forms of $\mathsf\Gamma_1$ and $\mathsf\Gamma_2$ are $\mathbf\Gamma_1=\begin{pmatrix}3 & -4\\4 & 3 \end{pmatrix} $
and $\mathbf\Gamma_2=\begin{pmatrix}3 & 4\\-4 & 3\end{pmatrix}$, respectively.
If there exist invertible matrices $\mathbf U$ and $\mathbf V$ such that $\mathbf U\mathbf\Gamma_1\mathbf V=\mathrm{diag}(a_1, a_2)$ and
$\mathbf U\mathbf\Gamma_2\mathbf V=\mathrm{diag}(b_1, b_2)$ for some non-zero integers $a_1$, $a_2$, $b_1$, $b_2$, then
$$ \mathbf V^{-1}\mathbf\Gamma_1^{-1}\mathbf U^{-1}\mathbf U\mathbf\Gamma_2\mathbf V = \mathbf V^{-1}\mathbf\Gamma_1^{-1}\mathbf\Gamma_2\mathbf V =\mathrm{diag}\left(\frac{b_1}{a_1}, \frac{b_2}{a_2}\right).$$
Hence, $\mathbf\Gamma_1^{-1}\mathbf\Gamma_2$ has two real eigenvalues, $\frac{b_1}{a_1}$ and $\frac{b_2}{a_2}$. This contradicts the fact that $\mathbf\Gamma_1^{-1}\mathbf\Gamma_2 = \begin{pmatrix}-\frac{7}{25} & \frac{24}{25}\\ -\frac{24}{25}& -\frac{7}{25} \end{pmatrix}$ has two complex eigenvalues, $-\frac{7}{25}+\frac{24}{25} \mathrm i$ and $-\frac{7}{25}-\frac{24}{25} \mathrm i$.

Note that the above process requires solving for the modular multiplicative inverse of the Gaussian integer. Since the ring of Gaussian integers is a Euclidean domain \cite{Dummit_2004}, for any $\mathsf M, \mathsf N \in \mathbb{Z}[\mathrm{i}]$ with $\mathsf M \neq 0$, there exist $\mathsf q, \mathsf r \in \mathbb{Z}[\mathrm{i}]$ such that $\mathsf N = \mathsf q\mathsf M + \mathsf r$, where $|\mathsf r| < |\mathsf M|$. It makes sense that the Euclidean algorithm can be used to find modular multiplicative inverses, similar to how it is used for integers. However, using (\ref{mathsfrN}) to compute $\mathsf r$ is insufficient. For example, if we let $\mathsf N = 5 + 10\mathrm{i}$ and $\mathsf M = 4 + 4\mathrm{i}$, then $\mathsf r = \mathsf N - \mathsf M\left\lfloor \frac{\mathsf N}{\mathsf M} \right\rfloor = 1 + 6\mathrm{i}$. It is evident that the condition $|\mathsf r| < |\mathsf M|$ is not satisfied.
To use the Euclidean algorithm for complex numbers, we introduce the following rounding operation:
$$ \left[\mathsf z \right] = \left[z_1 \right] + \left[z_2 \right]\mathrm i,$$
where $\left[z_i \right]$ satisfy
$-\frac{1}{2} \le z_i - \left[z_i \right] < \frac{1}{2}$ for $ i = 1, 2$.
It is easy to verify that if $\mathsf r = \mathsf N-\mathsf M\left[\frac{\mathsf N}{\mathsf M}\right]$, then $|\mathsf r| < |\mathsf M|$. Hence, we can recursively obtain the modular inverse of the Gaussian integer using the Euclidean algorithm. For example, we let $\mathsf n = 19+8\mathrm i$ and $\mathsf m = 3+4\mathrm{i}$. Clearly, $\left[\frac{\mathsf n}{\mathsf m} \right] = 4-2\mathrm{i}$. Hence,
$\mathsf n = (4-2\mathrm{i})\mathsf m +(-1-2\mathrm{i})$.
Since $\left[\frac{\mathsf m}{-1-2\mathrm{i}} \right] = -2$, we have
$\mathsf m = -2(-1-2\mathrm{i})+1$.
Then, we have the following Bezout's identity:
$$1 = 2\mathsf n + (-7+4\mathrm{i})\mathsf m.$$
Hence, the modular multiplicative inverse of $\mathsf m$ modulo $\mathsf n$ is $-7+4\mathrm{i}$.

\begin{example}\label{ex_C_CRT}
Let us consider the following system of congruence equations
\begin{equation*}\label{ex_mod_eq}
\begin{cases}
\mathsf N \equiv -3+6\mathrm i \mod 2(1+4\mathrm{i}), \\
\mathsf N \equiv -1-6\mathrm i \mod 2(-3-4\mathrm{i}), \\
\mathsf N \equiv -15+44\mathrm i \mod 2(13+16\mathrm i).
\end{cases}
\end{equation*}
Clearly, we have $\mathsf\Gamma_1 = 1+4\mathrm{i}$, $\mathsf\Gamma_2 = -3-4\mathrm{i}$, $\mathsf\Gamma_3 = 13+16\mathrm i$, and $M = 2$. Hence,
$\mathsf\gamma_1 = 25-100\mathrm i$, $\mathsf\gamma_2 = -51+68\mathrm i$, and $\mathsf\gamma_3 = 13-16\mathrm i$. By the Euclidean algorithm, we have
$\bar{\mathsf\gamma}_1 = 1$, $\bar{\mathsf\gamma}_2 = -2-2\mathrm i$, and $\bar{\mathsf\gamma}_3 =9+6\mathrm i$. By (\ref{eq_rc}), we have $\mathsf r^c = \left\langle \mathsf r_i \right\rangle_M = 1$, where $i=1, 2, 3$. By (\ref{eq_q_iNN0}), we have
$\mathsf q_1 = -2 + 3\mathrm i$, $\mathsf q_2 =-1-3\mathrm i$, and $\mathsf q_3 = -8+22\mathrm i$. Hence, we obtain by (\ref{N_0_mod}) that $\mathsf N_0 = 8+9\mathrm i$. It follows from (\ref{eq_N_r_c}) that $\mathsf N = 17+18\mathrm i$.
\end{example}

\subsection{Problem Description and SR-ADCs}

\begin{figure}[h]
\centering
\captionsetup{skip=-10pt}
\setlength{\belowcaptionskip}{-0.2cm}
\captionsetup{font={footnotesize}}
\captionsetup{labelsep=period}
\includegraphics[width=0.5\textwidth]{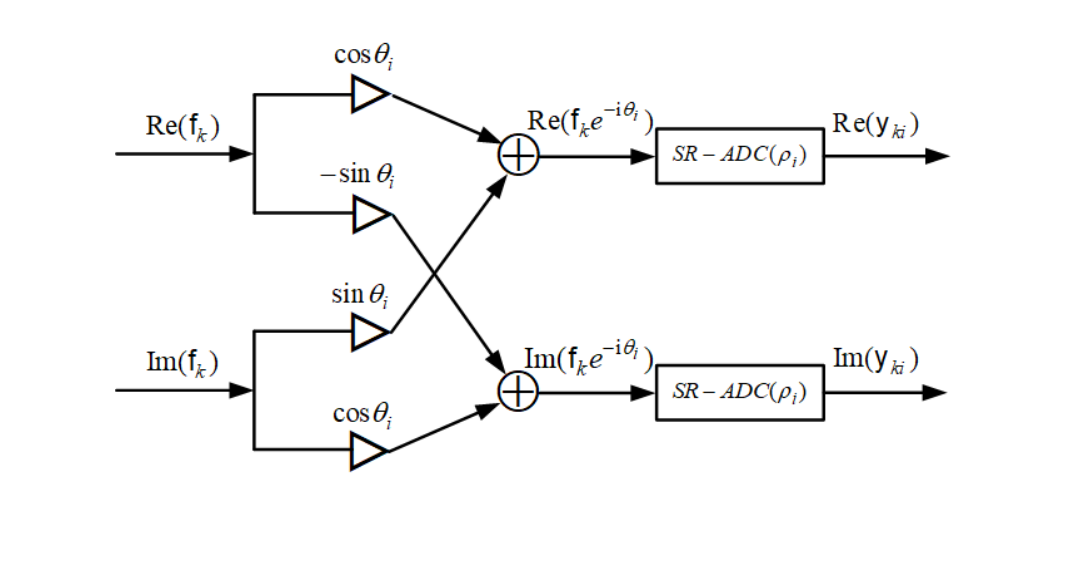}
\setlength{\belowcaptionskip}{-0.2cm}
\caption{Complex-valued modulus SR-ADCs \cite{Gong_2021}.}\label{fig:C_SR_ADC}
\end{figure}

Now we consider the problem of recovery of sampled values for complex-valued bandlimited signals using the C-CRT. For a single SR-ADC, one can compute the modulus of real numbers. Two combined SR-ADCs can obtain a modulus sampler with Gaussian integer $\mathsf M_i$,
as illustrated in Fig. \ref{fig:C_SR_ADC}. For convenience, we denote $\mathsf M_i = \rho_i e^{\mathrm{i}\theta_i}$, where $\theta_i$ represents the angle of $\mathsf M_i$, and $\rho_i$ represents the dynamic range of the SR-ADC. Let $T$ be the sampling interval length of each SR-ADC and $\mathsf f_k = \mathsf f(kT)$, where $\mathsf f(t)$ is a complex-valued bandlimited signal and $k \in \mathbb{Z}$. Then, the output $\mathsf y_{ki}$ can be expressed as
\begin{equation}\label{eq:f_k_n_ki}
\mathsf y_{ki} = \left\langle \mathsf f_k e^{-\mathrm i\theta_i} \right\rangle_{\rho_i}, \ i = 1, 2, \ldots, L.
\end{equation}
By applying a phase shift, one can obtain $\mathsf y_{ki}$ through separately applying the modulo operation to its real and imaginary parts. To be specific, if we rewrite $\mathsf y_{ki}$ as
$$\mathsf y_{ki} = \left\langle \mathrm{Re}(\mathsf f_ke^{-\mathrm{i}\theta_i}) \right\rangle_{\rho_i} + \mathrm i \left\langle \mathrm{Im}(\mathsf f_ke^{-\mathrm{i}\theta_i}) \right\rangle_{\rho_i},$$
then $\left\langle \mathrm{Re}(\mathsf f_{k}e^{-\mathrm{i}\theta_i}) \right\rangle_{\rho_i}$
and $\left\langle \mathrm{Im}(\mathsf f_{k}e^{-\mathrm{i}\theta_i}) \right\rangle_{\rho_i}$ can be obtained by
$$\left\langle \mathrm{Re}(\mathsf f_{k}e^{-\mathrm{i}\theta_{i}}) \right\rangle_{\rho_i} = \left\langle \mathrm{Re}(\mathsf f_{k})\cos\theta_i + \mathrm{Im}(\mathsf f_{k})\sin\theta_i\right\rangle_{\rho_i}$$
and $$\left\langle\mathrm{Im}(\mathsf f_{k}e^{-\mathrm{i}\theta_{i}})\right\rangle_{\rho_i} = \left\langle\mathrm{Im}(\mathsf f_{k})\cos\theta_i - \mathrm{Re}(\mathsf f_{k})\sin\theta_i\right\rangle_{\rho_i},$$
respectively. By (\ref{eq:f_k_n_ki}), we can obtain the system of congruences
$$\mathsf{f}_k \equiv \mathsf{y}_{ki}e^{\mathrm{i}\theta_i} \mod \mathsf{M}_i, \ i = 1, 2, \ldots, L.$$
For convenience, we let $\mathsf{N} = \mathsf{f}_k$, and let $\mathsf{r}_i = \mathsf{y}_{ki}e^{\mathrm{i}\theta_i}$, then we have
$$\mathsf{N} \equiv \mathsf{r}_i \mod \mathsf M_i,\ i = 1, 2, \ldots, L.$$
When there is no error in any remainder, $\mathsf f_k$ or $\mathsf{N}$ can be recovered by the C-CRT. However, in practical applications, the obtained remainders $\mathsf{r}_i$ may have errors. Let the erroneous remainders be
$$\tilde{\mathsf r}_i = \mathsf r_i + \Delta \mathsf r_i,$$
where $i = 1, 2, \ldots, L$, and $\Delta \mathsf r_i$ represents the error in the $i$-th remainder. In the following, we consider the MLE C-CRT based on the assumption of wrapped complex Gaussian distributions of the errors, which is to estimate the complex value $\mathsf{N}$ from its erroneous complex remainders $\tilde{\mathsf r}_1, \tilde{\mathsf r}_2, \ldots, \tilde{\mathsf r}_L$.

\section{MLE C-CRT and Its Fast Algorithm}\label{sec_MLE}
In this section, we first introduce circular distance and wrapped distributions. Then, we provide the expression of the MLE C-CRT. Finally, we propose a fast algorithm for the MLE C-CRT.

\subsection{Circular Distance and Wrapped Distributions}
First, we introduce the definition of circular distance for complex numbers. For two complex numbers $\mathsf x$ and $\mathsf y$, and a nonzero Gaussian integer $\mathsf M$, we define the circular distance between $\mathsf x$ and $\mathsf y$ for $\mathsf M$ as
\begin{equation}\label{d_mathsf_M}
d_\mathsf M(\mathsf x, \mathsf y) = \mathsf x - \mathsf y - \left[\frac{\mathsf x - \mathsf y}{\mathsf M} \right] \mathsf M.
\end{equation}
The circular distance has the following properties, which are proven in Appendix $A$.
\begin{property}\label{prop1_d_M}
$d_\mathsf M(\mathsf x, \mathsf y)\in \mathcal S_\mathsf M$, where
\begin{equation}\label{eq_def_SM}
\mathcal S_\mathsf M = \left\{\mathsf M(c + d \mathrm i): -\frac{1}{2}\leq c, d < \frac{1}{2}\right\}.
\end{equation}
\end{property}

Similar to Property \ref{prop2_F_M}, we have that $\mathcal S_\mathsf M$ is a square with side length $|\mathsf M|$.

\begin{property}\label{prop2_d_M}
For any Gaussian integer $\mathsf k$, it holds that $d_\mathsf{\mathsf M}(\mathsf x +\mathsf {kM}, \mathsf y)=d_\mathsf{\mathsf M}(\mathsf x, \mathsf y+\mathsf {kM}) = d_\mathsf{\mathsf M}(\mathsf x, \mathsf y)$.
Furthermore, $d_{\mathsf M}(\mathsf x, \mathsf y)=d_{\mathsf M}(\mathsf x , \left\langle \mathsf y\right\rangle_\mathsf M)$.
\end{property}

%$$d_\mathsf{\mathsf M}(\mathsf x +\mathsf {kM}, \mathsf y) = left[\frac{\mathsf x+\mathsf{kM} -\mathsf y}{\mathsf M}\right] =\left[\frac{\mathsf x-\mathsf y}{\mathsf M}\right] +\mathsf k, \left[\frac{\mathsf x-\mathsf y-\mathsf{kM}}{\mathsf M}\right]=\left[\frac{\mathsf x-\mathsf y}{\mathsf M}\right] -\mathsf k.$$

\begin{property}\label{prop3_d_M}
If $\mathsf x-\mathsf y\in \mathcal S_{\mathsf M}$, then $d_\mathsf{M}(\mathsf x, \mathsf y)=\mathsf x-\mathsf y$.
\end{property}

\begin{property}\label{prop4_d_M}
If $\mathsf x-\mathsf y \in \partial\mathcal S_\mathsf M$,
where $\partial\mathcal S_\mathsf M$ denotes the boundary of $\mathcal S_\mathsf M$, then $|d_{\mathsf M}(\mathsf x, \mathsf y)|=|\mathsf x-\mathsf y|$.
\end{property}

\begin{property}\label{prop5_d_M}
Let $k$ be a nonzero integer. If $|\mathsf M|\geq \sqrt{2}$, then
$|d_{k\mathsf M}(d_k(\mathsf x, \mathsf y), 0)|=|d_k(\mathsf x, \mathsf y)|$.
\end{property}

Fig. \ref{S_M_1_2} gives an illustration of circular distance of $\mathsf x = 3+3\mathrm i$ and $\mathsf y = 1-2\mathrm i$. By the definition of the circular distance, we have $d_{3+4\mathrm i}(\mathsf x, \mathsf y) = -1-\mathrm i$ and $d_{4}(\mathsf x, \mathsf y) = -2+\mathrm i$.

\begin{figure}[h]
\centering
\setlength{\belowcaptionskip}{-0.2cm}
\setlength{\subfigcapskip}{-10pt}
\captionsetup{font={footnotesize}}
\captionsetup{labelsep=period}
\subfigure[$\mathsf M = 3+4\mathrm i$.]{\includegraphics[width=.33 \textwidth]{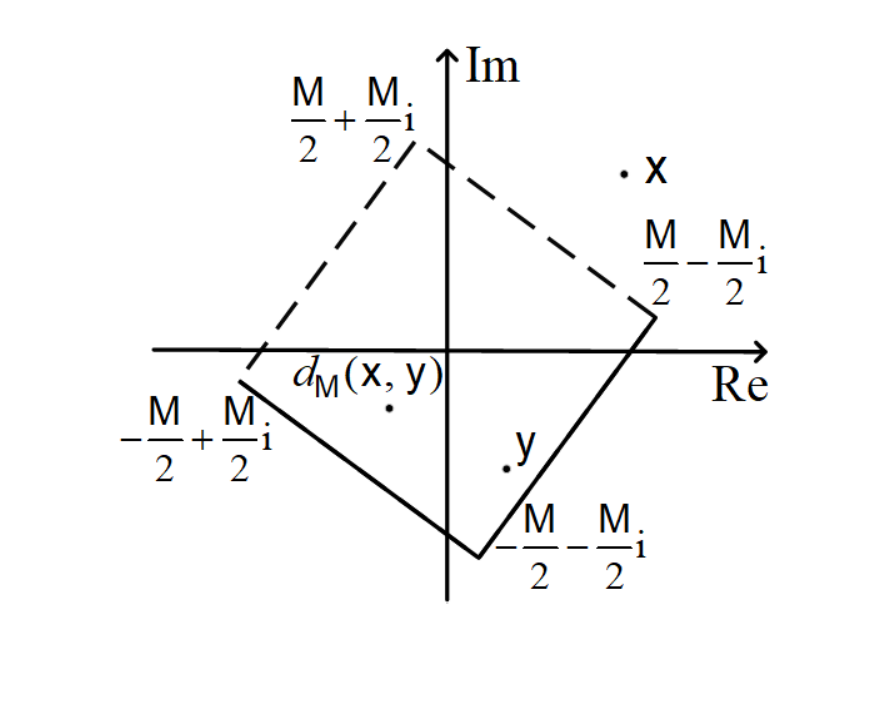}}
\subfigure[$\mathsf M = 4$.]{\includegraphics[width=.3 \textwidth]{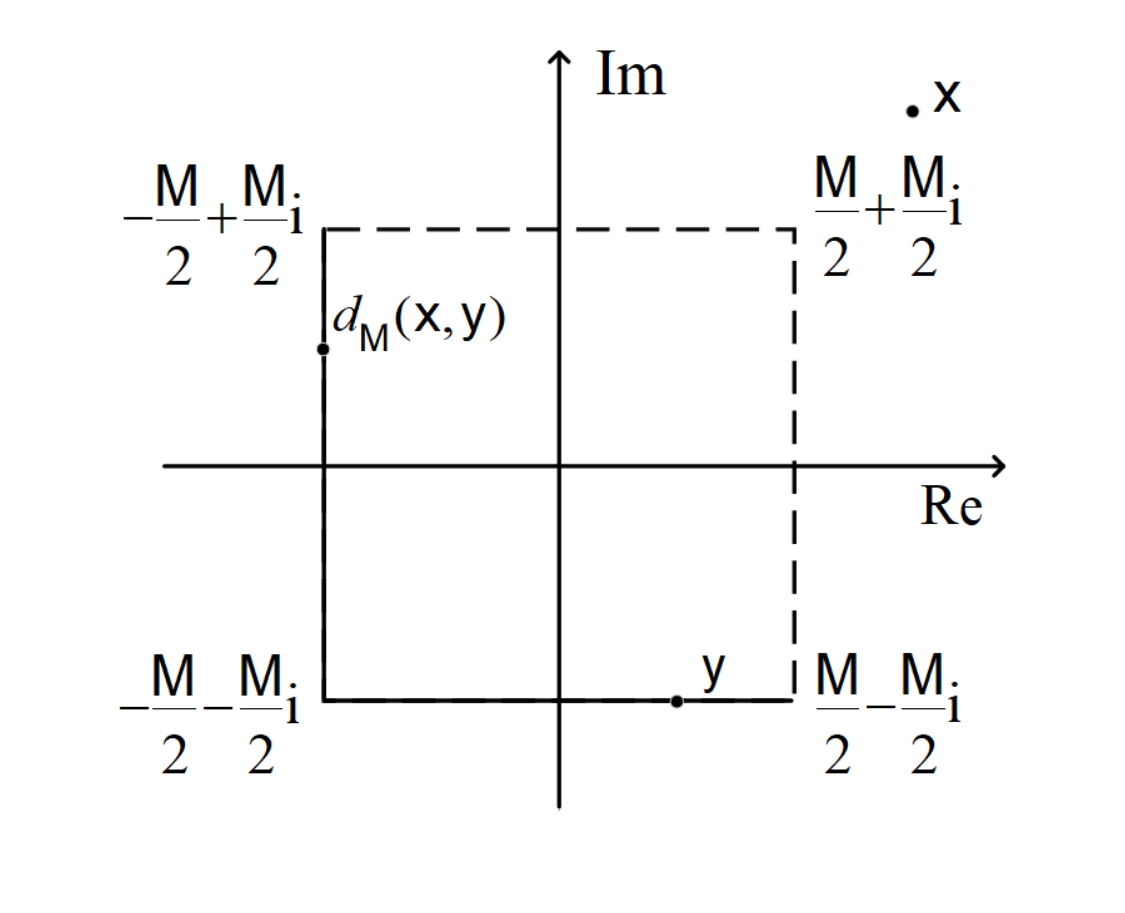}}
\setlength{\belowcaptionskip}{-0.2cm}
\caption{Illustration of circular distance based on $\mathcal S_\mathsf M$.}\label{S_M_1_2}
\end{figure}

\begin{proposition}\label{Th_du_dx}
If $d_\mathsf M(\mathsf r, \mathsf N)$ is considered as a function of $\mathsf r=x+y\mathrm i \in \mathcal F_\mathsf M$, where $\mathsf M =\rho e^{\mathrm i\theta}$ is a Gaussian integer, and $\mathsf N$ is a complex number. Then, the points of discontinuity of $d_\mathsf M(\mathsf r, \mathsf N)$ belong to the following set:
$$\mathcal D = \{\mathsf r: y\sin\theta+x\cos\theta = c_1\text{\ or\ } y\cos\theta-x\sin\theta = c_2\},$$
where $c_1=\pm \frac{\rho}{2}+\mathrm{Re}\left(\langle \mathsf N e^{-\mathrm i\theta}\rangle_{\rho} \right)$ and
$c_2 = \pm \frac{\rho}{2}+\mathrm{Im}\left(\langle \mathsf N e^{-\mathrm i\theta}\rangle_{\rho} \right)$. Furthermore, the measure of $\mathcal D$ is zero.
\end{proposition}

\begin{proof}
Note that $d_\mathsf M(\mathsf r, \mathsf N) = d_\mathsf M(\mathsf r, \left\langle \mathsf N \right\rangle_\mathsf M)$ by Property \ref{prop2_d_M}, we consider $\mathsf N\in\mathcal F_\mathsf M$ without loss of generality. Since $\mathsf r, \mathsf N \in \mathcal F_\mathsf M$, we have $\mathsf re^{-\mathrm i\theta}, \mathsf Ne^{-\mathrm i\theta}\in \mathcal F_{\rho}$
by Property \ref{prop1_F_M}. Hence,
$\mathrm{Re}\left(\mathsf re^{-\mathrm i\theta}-\mathsf Ne^{-\mathrm i\theta}\right), \mathrm{Im}\left(\mathsf re^{-\mathrm i\theta}- \mathsf Ne^{-\mathrm i\theta}\right)\in (-\rho, \rho)$. Note that
$$d_\mathsf M(\mathsf r, \mathsf N)=\mathsf r - \mathsf N - \left[\frac{\mathsf re^{-\mathrm i\theta} - \mathsf Ne^{-\mathrm i\theta}}{\rho} \right] \mathsf M.$$
That is,
$$d_\mathsf M(\mathsf r, \mathsf N) = \mathsf r - \mathsf N - \left[\frac{\mathrm{Re}(\mathsf re^{-\mathrm i\theta} - \mathsf Ne^{-\mathrm i\theta})}{\rho} \right] \mathsf M - \mathrm i\left[\frac{\mathrm{Im}(\mathsf re^{-\mathrm i\theta} - \mathsf Ne^{-\mathrm i\theta})}{\rho} \right] \mathsf M.$$
Thus, $d_\mathsf M(\mathsf r, \mathsf N)$ is discontinuous at $\mathsf r$ when the real or imaginary part of $\mathsf re^{-\mathrm i\theta} - \mathsf Ne^{-\mathrm i\theta}$ equals $\frac{\rho}{2}$ or $-\frac{\rho}{2}$.
Note that since the proofs of these four cases are similar, we only consider the case when $\mathrm{Re}\left(\mathsf re^{-\mathrm i\theta}-\mathsf Ne^{-\mathrm i\theta}\right) = \frac{\rho}{2}$, i.e., $\mathrm{Re}\left(\mathsf re^{-\mathrm i\theta}\right) = \frac{\rho}{2}+\mathrm{Re}\left(\mathsf Ne^{-\mathrm i\theta}\right)$. Since $$\mathsf re^{-\mathrm i\theta} = (y\sin\theta+x\cos\theta) + (y\cos\theta-x\sin\theta)\mathrm i,$$ we have $$y\sin\theta+x\cos\theta=\frac{\rho}{2}+\mathrm{Re}\left(\mathsf N e^{-\mathrm i\theta}\right).$$
Hence, $\mathsf r \in \mathcal D$. Since $\mathcal D$ contains at most four line segments, its measure is zero.
\end{proof}

We now consider that the unknown complex number $\mathsf N$ to determine is noisy in its observation and the noise follows a complex Gaussian distribution as commonly assumed. This noise results in noises in the remainders. To study the noises in the remainders, we introduce the multivariate wrapped distributions below. Let $f_{\mathbf X}(\mathbf x)$ be the probability density function (pdf) of a 2D random variable $\mathbf X$.
As described in \cite{Greco_2021}, the pdf of $\mathbf Y \equiv \mathbf X \mod \mathbf I$ is
$$f_{\mathbf Y}(\mathbf y)=\sum_{\mathbf k\in \mathbb Z^2}f_{\mathbf X}(\mathbf y+ \mathbf k),$$
where $\mathbf y \in \mathrm{FPD}(\mathbf I)$, and $\mathbf I$ is the identity matrix. Next, we consider the pdf of $\mathbf Z\equiv \mathbf X \mod \mathbf M$, where $\mathbf M$ is an invertible matrix. Clearly,
$$\mathbf M^{-1}\mathbf Z\equiv \mathbf M^{-1}\mathbf X \mod \mathbf I.$$
If we let $\mathbf Z'=\mathbf M^{-1}\mathbf Z$, then
$$f_{\mathbf Z'}(\mathbf z')=|\det(\mathbf M)|\sum_{\mathbf k\in \mathbb Z^2}f_{\mathbf X}(\mathbf M(\mathbf z'+ \mathbf k)),$$
where $\det(\mathbf M)$ is the determinant of $\mathbf M$. Consequently, the pdf of $\mathbf Z$ is
\begin{equation}\label{eq_pdf_of_fZ}\hspace{-0em}
f_{\mathbf Z}(\mathbf z)=|\det(\mathbf M^{-1})|f_{\mathbf Z'}(\mathbf M^{-1}\mathbf z) = \sum\limits_{\mathbf k\in \mathbb Z^2}f_{\mathbf X}(\mathbf z+\mathbf M \mathbf k),
\end{equation}
where $\mathbf z \in \mathrm{FPD}(\mathbf M)$. Then, (\ref{eq_pdf_of_fZ}) can be used to represent the pdf of the complex random variable $\mathsf X = \mathsf N+\mathsf W$, where $\mathsf W$ is the noise of $0$ mean and $\mathsf N$ is the true unknown complex number to determine. Specifically, if $\mathsf W$ follows a complex Gaussian distribution, then the pdf of $\mathsf X$ is
$$f_{\mathsf X}(\mathsf x) = \frac{1}{2\pi \sigma^2}\exp \left\{-\frac {\left|{\mathsf x} - \mathsf N \right|^2} {2\sigma^2} \right\},$$
where $\sigma^2$ is the variance of both the real and imaginary parts of $\mathsf W$.
Thus, from (\ref{eq_pdf_of_fZ}) we have the pdf of $\mathsf R\equiv \mathsf X\mod \mathsf M$:
\begin{equation}\label{eq_pdf_of_R}
f_{\mathsf R}(\mathsf r) = \frac{1}{2\pi \sigma^2}\sum\limits_{\mathsf k \in \mathbb Z[\mathrm i]} \exp \left\{-\frac {\left|{\mathsf r} - \mathsf N + \mathsf k \mathsf M \right|^2} {2\sigma^2} \right\},
\end{equation}
where $\mathsf r\in \mathcal F_{\mathsf M}$. Let $\mathsf k' = \mathsf k + \left[\frac{\mathsf r -\mathsf N}{\mathsf M}\right]$, we can obtain from (\ref{eq_pdf_of_R}) that
\begin{align}\label{eq:pdf_hat_ri2}
f_{\mathsf R}(\mathsf r)
& =\frac{1}{2\pi\sigma^2}
\sum\limits_{\mathsf k' \in \mathbb Z[\mathrm i]} \exp \left
\{-\frac{|\mathsf r - \mathsf N - \left[\frac{\mathsf r - \mathsf N}{\mathsf M}\right]\mathsf M + \mathsf k'\mathsf M|^2} {2\sigma^2} \right\}\nonumber\\
& =  \frac{1}{2\pi\sigma^2} \sum\limits_{\mathsf k' \in \mathbb Z[\mathrm i]} \exp \left
\{-\frac{|d_\mathsf M(\mathsf r, \mathsf N) + \mathsf k'\mathsf M|^2} {2\sigma^2} \right\}.
\end{align}

To simplify the expression of (\ref{eq:pdf_hat_ri2}), we introduce a proposition below.

\begin{proposition}\label{Th_3sigma}
Let $\mathsf r=x+y\mathrm i\in\mathcal F_{\mathsf M}$. If $|\mathsf M| \geq 6\sqrt{2}\sigma$, then
\begin{equation}\label{eq_3sigma}
\frac{1}{2\pi\sigma^2}\iint_{\mathcal F_{\mathsf M}}\exp\left
\{-\frac{|d_\mathsf M(\mathsf r, \mathsf N)|^2}{2\sigma^2} \right\}\mathrm dx\mathrm dy > 0.9973^2.
\end{equation}
\end{proposition}

\begin{proof}
By Proposition \ref{Th_du_dx}, we know that the discontinuity points of $d_{\mathsf M}(\mathsf r, \mathsf N)$ in $\mathcal F_{\mathsf M}$ form a set with measure zero.
Thus, $g(\mathsf r) \triangleq \exp\left\{-\frac{|d_{\mathsf M}(\mathsf r, \mathsf N)|^2}{2\sigma^2} \right\}$ is integrable on $\mathcal F_{\mathsf M}$.
We divide $\mathcal F_{\mathsf M}$ into $n$ subrectangles equally for convenience, say $\mathcal P_1, \mathcal P_2, \ldots, \mathcal P_n$. For any $\mathsf p_j\in\mathcal P_j$ such that $d_\mathsf M(\mathsf r, \mathsf N)$ is continuous at $\mathsf p_j$, denoting $d_{\mathsf M}(\mathsf p_j, \mathsf N) = u_j + v_j\mathrm i$, we have
\begin{equation}\label{eq_int_FM}
\begin{aligned}
 \iint_{\mathcal F_{\mathsf M}}\exp\left\{-\frac{|d_{\mathsf M}(\mathsf r, \mathsf N)|^2}{2\sigma^2} \right\}\mathrm dx\mathrm dy
 =\lim_{n\to\infty}\frac{|\mathsf M|^2}{n} \sum_{j=1}^{n}\exp\left\{-\frac{ u_j^2+v_j^2}{2\sigma^2} \right\}
 =\iint_{\mathcal S_{\mathsf M}}\exp\left\{-\frac{u^2+v^2}{2\sigma^2} \right\}\mathrm du\mathrm dv.
\end{aligned}
\end{equation}
Let $T=\frac{\sqrt{2}}{2}|\mathsf M|$. According to (\ref{eq_def_SM}), we know that $\mathcal S_{T}$ is a square inscribed in the incircle of $\mathcal S_{\mathsf M}$, as shown in Fig. \ref{STSM}. Hence,
\begin{align}\label{iint_SM}
\left(\int_{-\frac{T}{2}}^{\frac{T}{2}}\exp\Big\{-\frac{u^2}{2\sigma^2}
 \Big\}\mathrm du\right)^2 = \iint_{\mathcal S_T}\exp\left\{-\frac{u^2+v^2}{2\sigma^2} \right\}\mathrm du\mathrm dv
  <\iint_{\mathcal S_{\mathsf M}}\exp\left\{-\frac{u^2+v^2}{2\sigma^2} \right\}\mathrm du\mathrm dv.
\end{align}
Note that $$\frac{1}{\sqrt{2\pi} \sigma}\int_{-3\sigma}^{3\sigma}\exp\left\{-\frac{u^2}{2\sigma^2} \right\}\mathrm du = 0.9973.$$ Since $|\mathsf M|\geq 6\sqrt{2}\sigma$, i.e., $T \ge 6\sigma$, we have
$$ \frac{1}{2\pi\sigma^2}\left(\int_{-\frac{T}{2}}^{\frac{T}{2}}\exp\left\{-\frac{u^2}{2\sigma^2} \right\}\mathrm du\right)^2 \ge 0.9973^2.$$
By (\ref{iint_SM}), we have $$\frac{1}{2\pi\sigma^2}\iint_{\mathcal S_{\mathsf M}}\exp\left\{-\frac{u^2+v^2}{2\sigma^2} \right\}\mathrm du\mathrm dv > 0.9973^2.$$
This leads to (\ref{eq_3sigma}) by (\ref{eq_int_FM}) .
\end{proof}

\begin{figure}[htbp]
\centering
\captionsetup{skip=-2pt}
\captionsetup{font={footnotesize}}
\captionsetup{labelsep=period}
\includegraphics[width=0.3\textwidth]{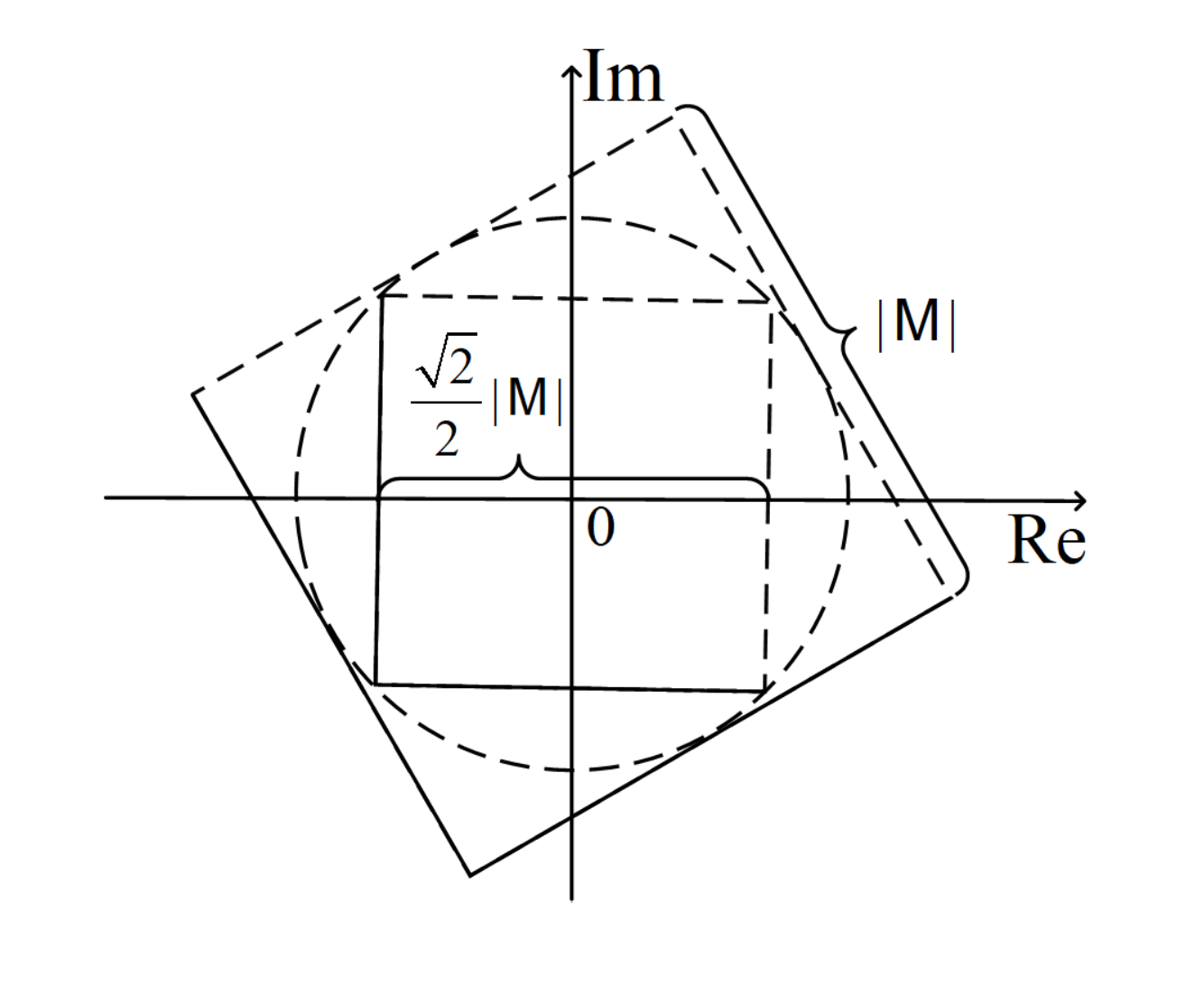}
\setlength{\belowcaptionskip}{-0.2cm}
\caption{Description of $\mathcal S_T$ and $\mathcal S_{\mathsf M}$.}\label{STSM}
\end{figure}

By Proposition \ref{Th_3sigma}, when $|\mathsf M|\geq 6\sqrt{2}\sigma$,
the pdf (\ref{eq:pdf_hat_ri2}) can be approximated as
\begin{equation}\label{eq_approximated_pdf}
f_{\mathsf R}(\mathsf r)\approx \frac{1}{2\pi\sigma^2}\exp\left
\{-\frac{|d_\mathsf M(\mathsf r, \mathsf N)|^2}{2\sigma^2} \right\}.
\end{equation}

\subsection{MLE C-CRT}
Now, we calculate the maximum likelihood function for $\mathsf N\in\mathcal F_{M\Gamma}$, which is the fixed but unknown complex number to determine as mentioned above. Assume that $\mathsf X_i=\mathsf N+\mathsf W_i$ follow complex Gaussian distributions, and $\mathsf W_i$ are independent of each other, where $i=1, 2, \ldots, L$. The variances of the real and imaginary parts of $\mathsf W_i$ are both $\sigma_i^2$, and their means are both $0$ for each $i$. Denote $\mathsf R_i$ as the random variable with observed value $\tilde{\mathsf r}_i$, which satisfies $\mathsf R_i \equiv \mathsf X_i\mod M\mathsf\Gamma_i$ and $\Gamma=\prod_{i=1}^L\mathsf\Gamma_i$ is an integer as assumed earlier. Then, $\mathsf R_i$ follows a wrapped complex Gaussian distribution as studied above.
For each $i$, from (\ref{eq_pdf_of_R}), we can obtain the conditional pdf of $\mathsf R_i$ as follows:
$$ f_{\mathsf R_i} (\tilde{\mathsf r}_i\mid \mathsf N ) = \frac{1}{2\pi \sigma_i^2}\sum\limits_{\mathsf k \in \mathbb Z[\mathrm i]} \exp \left\{-\frac {\left|\tilde{\mathsf r}_i - \mathsf N+ \mathsf k M \mathsf\Gamma_i \right|^2} {2\sigma_i ^2} \right\}.$$
Generally, $|M\mathsf\Gamma_i|$ is much larger than $\sigma_i$.
By Proposition \ref{Th_3sigma} and (\ref{eq_approximated_pdf}), we have
$$f_{\mathsf R_i}(\tilde{\mathsf r}_i\mid\mathsf N)\approx \frac{1}{2\pi\sigma_i^2}\exp\left
\{-\frac{|d_{M\mathsf\Gamma_i}(\tilde{\mathsf r}_i, \mathsf N)|^2}{2\sigma_i ^2} \right\}.$$
Consequently, we can approximate the joint conditional pdf
$f_{\mathsf R_1, \mathsf R_2, \ldots, \mathsf R_L}(\tilde{\mathsf r}_1, \tilde{\mathsf r}_2, \ldots, \tilde{\mathsf r}_L\mid\mathsf N) = \prod_{i=1}^L f_{\mathsf R_i}(\tilde{\mathsf r}_i\mid\mathsf N)$ as
$$ (2\pi)^{-L}\prod_{i=1}^{L} \sigma_i^{-2} \exp \left\{\! -\sum_{i=1}^L {\frac{1} {2\sigma_i^2}|d_{M\mathsf\Gamma_i}\left(\tilde{\mathsf r}_i, \mathsf N\right)|^2}\right\}.$$
Therefore, we have the log likelihood function of $\mathsf N$
\begin{equation} \label{eq:log_likelihood_fun}
\mathcal{L}\left(\mathsf z \right) = -L \ln 2\pi - 2\sum_{i=1}^L \ln \sigma_i-
\sum_{i=1}^L\frac{|d_{M\mathsf\Gamma_i}\left(\tilde{\mathsf r}_i, \mathsf z \right)|^2}{2\sigma_i^2}.
\end{equation}
The MLE maximizes $\mathcal{L}(\mathsf z)$ with respect to an unknown complex number $\mathsf z \in \mathcal F_{M\Gamma}$, which yields the following minimization problem
\begin{equation}\label{eq:minmization_problem}
\begin{aligned}
 \hat{\mathsf N}_{\text{MLE}} & =  \arg \max\limits_{\mathsf z \in \mathcal F_{M\Gamma}} \mathcal{L}\left(\mathsf z\right)
 =   \arg \min\limits_{\mathsf z \in \mathcal F_{M\Gamma}} \sum_{i=1}^L {\frac{1} {\sigma_i ^2}|d_{M\mathsf\Gamma_i}\left(\tilde{\mathsf r}_i, \mathsf z\right)|^2}.
\end{aligned}
\end{equation}
\begin{figure}[h!]
\centering
\captionsetup{font={footnotesize}}
\captionsetup{labelsep=period}
\includegraphics[width= 0.5\textwidth]{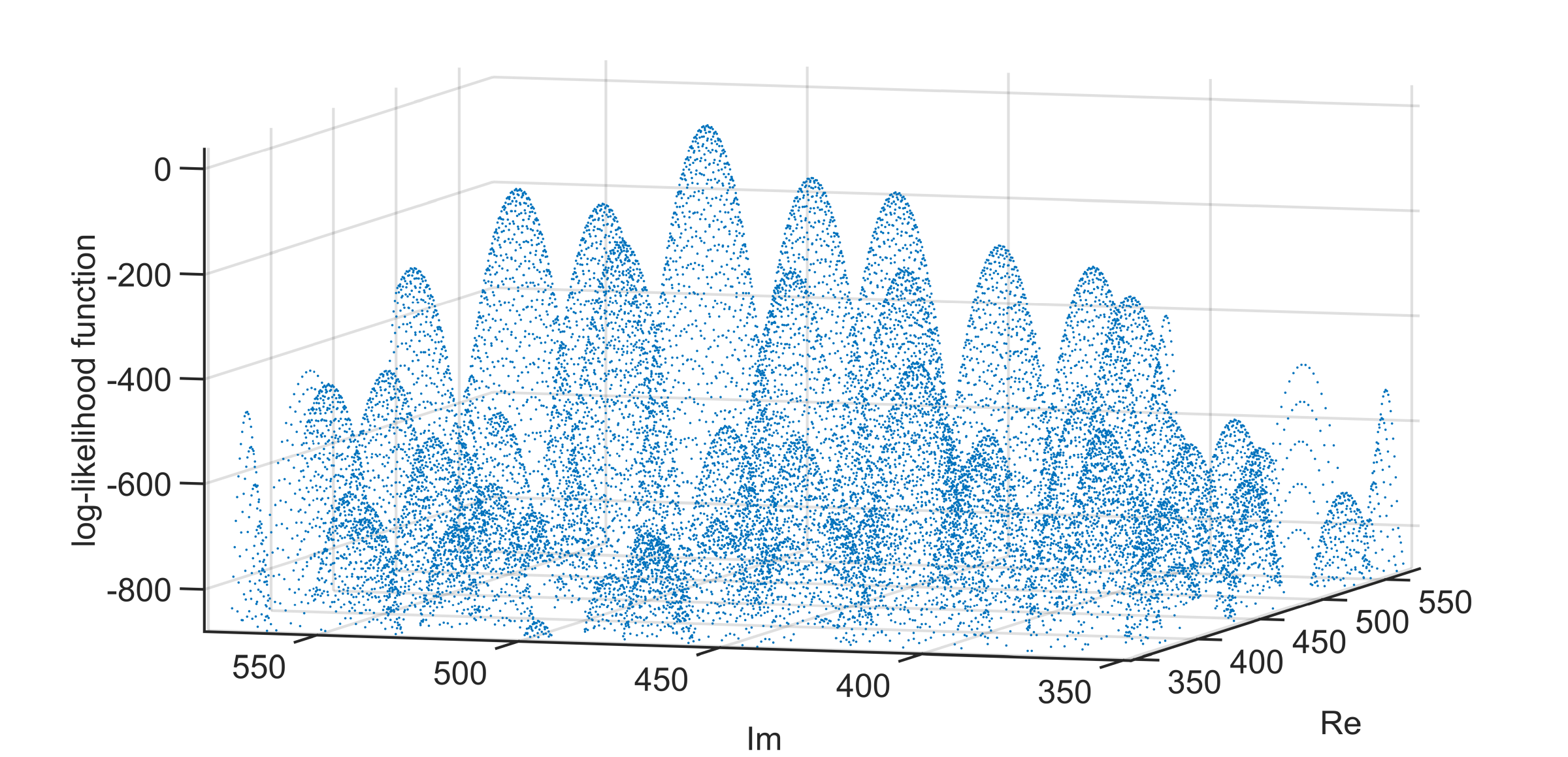}
\setlength{\belowcaptionskip}{-0.2cm}
\caption{The log likelihood function (\ref{eq:log_likelihood_fun}).}  \label{log_likelihood}
\end{figure}
In Fig. \ref{log_likelihood}, we show the right-hand side of the log likelihood function in (\ref{eq:log_likelihood_fun}), where $\mathsf N = 500 + 500\mathrm i$, $M=10$, $\mathsf\Gamma_1 = 3 + 4\mathrm i,  \mathsf\Gamma_2 = 3 - 4\mathrm i, \mathsf\Gamma_3 = 4$, and $\sigma_1$ to $\sigma_3$ are $0.2$, $0.3$, and $0.4$, respectively. In this case, $\Gamma=100$. By (\ref{eq:minmization_problem}), we have $\hat{\mathsf N}_\text{MLE} = 501 + 500\mathrm i$.

Note that $\mathsf z$ in (\ref{eq:minmization_problem}) may take any complex number in $\mathcal F_{M\Gamma}$. In general, solving the minimization problem (\ref{eq:minmization_problem}) may have a high computational complexity by searching the whole 2D region $\mathcal F_{M\Gamma}$. Next, we will present a fast algorithm with only $2L$ searches.

\subsection{Fast MLE C-CRT Algorithm}
From (\ref{eq_q_iNN0}), (\ref{N_0_mod}), and (\ref{eq_N_r_c}), one can see that the common remainder $\mathsf r^c$ is crucial for reconstructing $\mathsf N$. For noisy remainders $\tilde{\mathsf r}_i$ of $\mathsf N$ modulo $\mathsf M_i$, their remainders modulo $M$, i.e.,
$\tilde{\mathsf r}^c_i = \langle\tilde{\mathsf r}_i\rangle_M$
may be different from each other due to the errors for $i = 1, 2, \ldots, L$. To estimate the common remainder from $\tilde{\mathsf r}^c_1, \tilde{\mathsf r}^c_2, \ldots, \tilde{\mathsf r}^c_L$, we define a special averaging operation of $\tilde{\mathsf r}^c_i$ as
\begin{equation} \label{eq:define_har_rc}
\hat{\mathsf r}^c \triangleq \arg \min\limits_{\mathsf x \in \mathcal F_M} \sum_{i=1}^{L} \frac{1}{\sigma_i^2} |d_M\left(\tilde{\mathsf r}^c_i, \mathsf x\right)|^2.
\end{equation}
After the common remainder $\mathsf r^c$ is estimated above, we can estimate $\mathsf q_i$ as
\begin{equation} \label{eq:estimating_q_i}
\hat{\mathsf q}_i = \left[\frac{\tilde{\mathsf r}_i - \hat{\mathsf r}^c}{M}\right ], \ i=1, 2, \ldots, L.
\end{equation}
Consequently, $\mathsf N_0$ can be estimated by
\begin{equation}\label{eq:calculate_N_0}
\hat{\mathsf N}_0 = \left\langle \sum_{i=1}^L \bar {\mathsf\gamma}_i\mathsf \gamma_i\hat{\mathsf q}_i\right\rangle_\Gamma.
\end{equation}
Therefore, $\mathsf N$ can be estimated by
\begin{equation} \label{eq:calculate_N2B}
\hat{\mathsf N} = M\hat{\mathsf N}_0 + \hat{\mathsf r}^c.
\end{equation}

The following result says that the obtained $\hat{\mathsf N}$ in (\ref{eq:calculate_N2B}) is indeed the MLE when the estimate of $\mathsf r^c$ is $\hat{\mathsf r}^c$ in (\ref{eq:define_har_rc}).

%**************** **********************************************************
\begin{theorem}\label{th_MLE}
If $\mathsf N \in \mathcal F_{M\Gamma}$, then $\hat{\mathsf N}$ in (\ref{eq:calculate_N2B}) is the MLE of $\mathsf N$, that is, $\hat {\mathsf N} = \hat {\mathsf N}_{\text{MLE}}$.
\end{theorem}
\begin{proof}
By (\ref{eq:calculate_N_0}) and Theorem \ref{Th_C_CRT}, we have
$$\hat{\mathsf N}_0\equiv \hat{\mathsf q}_i\mod\mathsf\Gamma_i,\ i=1, 2, \ldots, L.$$
That is, there exist Gaussian integers $\mathsf k_i$ such that
$\hat{\mathsf N}_0=\mathsf k_i \mathsf \Gamma_i + \hat {\mathsf q}_i$.
By (\ref{eq:calculate_N2B}), we have
$$\hat{\mathsf N} = M\mathsf k_i \mathsf \Gamma_i + M\hat {\mathsf q}_i+\hat{\mathsf r}^c.$$
According to Property \ref{prop2_d_M}, we have
$$|d_{M \mathsf\Gamma_i}(\tilde{\mathsf r}_i, \hat{\mathsf N})|
=\left|d_{M \mathsf\Gamma_i}(\tilde{\mathsf r}_i, M \hat{\mathsf q}_i + \hat{\mathsf r}^c)\right|
=\left|d_{M \mathsf\Gamma_i}(\tilde{\mathsf r}_i - M \hat{\mathsf q}_i - \hat{\mathsf r}^c, 0)\right|.$$
It follows from (\ref{eq:estimating_q_i}) that
$$|d_{M \mathsf\Gamma_i}(\tilde{\mathsf r}_i, \hat{\mathsf N})|=\left|d_{M \mathsf\Gamma_i}(d_M(\tilde{\mathsf r}_i, \hat{\mathsf r}^c ), 0)\right|.$$
As $d_M\left(\tilde{\mathsf r}_i, \hat {\mathsf r}^c\right) \in \mathcal S_M$, by Property \ref{prop5_d_M}, we have
$|d_{M \mathsf\Gamma_i}(\tilde{\mathsf r}_i, \hat{\mathsf N})|=\left|d_M(\tilde{\mathsf r}_i, \hat{\mathsf r}^c)\right|$.
Hence, (\ref{eq:minmization_problem}) and (\ref{eq:define_har_rc}) are equivalent, that is, $\mathsf z$ in (\ref{eq:minmization_problem}) is optimal
if and only if $\mathsf x$ in (\ref{eq:define_har_rc}) is optimal.
\end{proof}

According to Theorem \ref{th_MLE}, we can search for $\hat{\mathsf r}^c$ within the smaller set $\mathcal F_M$ to obtain $\hat{\mathsf N}$. Although $\mathcal F_M$
is much smaller than the original searching region $\mathcal F_{M\Gamma}$ in (\ref{eq:minmization_problem}), it still contains infinitely many elements to search. Next, we introduce a fast algorithm that requires only a finite number of searches to find $\hat{\mathsf r}^c$. Let
\begin{equation}\label{eq_weights}
w_i = \frac{\frac{1}{\sigma_i^{2}}}{\sum_{i=1}^{L}\frac{1}{\sigma_i^{2}}}, \ i=1, 2, \ldots, L.
\end{equation}
Then, we have $0 < w_i \le 1$ and $\sum_{i=1}^{L}w_i=1$. Consequently, the optimal estimate in (\ref{eq:define_har_rc}) can be rewritten as
\begin{equation}\label{estimate_rc_1}
\hat{\mathsf r}^c = \arg \min\limits_{\mathsf x \in \mathcal F_M} \sum_{i=1}^{L}w_i|d_M
\left(\tilde{\mathsf r}_i^c, \mathsf x\right)|^2.
\end{equation}

\begin{theorem}\label{Hat_rcImage}
The estimate $\hat{\mathsf r}^c$ in (\ref{estimate_rc_1}) is optimal if and only if
$\mathrm{Re}(\hat{\mathsf r}^c)$ and $\mathrm{Im}(\hat{\mathsf r}^c)$ are optimal simultaneously, i.e.,
\begin{equation}\label{Re_rc_Im_rc}\hspace{-1em}\left\{
\begin{aligned}
& \mathrm{Re}(\hat{\mathsf r}^c) = \arg \min\limits_{\mathrm{Re}(\mathsf x) \in [0, M)} \sum_{i=1}^{L}w_i d_M^2\left(\mathrm{Re}(\tilde{\mathsf r}_i^c), \mathrm{Re}(\mathsf x)\right), \\
& \mathrm{Im}(\hat{\mathsf r}^c) = \arg \min\limits_{\mathrm{Im}(\mathsf x) \in [0, M)}\sum_{i=1}^{L}w_i d_M^2\left(\mathrm{Im}(\tilde{\mathsf r}_i^c), \mathrm{Im}(\mathsf x)\right).
\end{aligned}\right.
\end{equation}
\end{theorem}

\begin{proof}
For any $\mathsf x\in \mathcal F_M$, we have
$$\hspace{-0.8em} \begin{aligned}
& \sum_{i=1}^L w_i|d_M \left(\tilde{\mathsf r}_i^c, \mathsf x\right)|^{2} =
\sum_{i=1}^L w_i\left|\tilde{\mathsf r}_i^c - \mathsf x - \left[\frac{\tilde{\mathsf r}_i^c - \mathsf x}{ M}\right]\right|^{2}\\
& = \sum_{i=1}^{L}w_i \left(\left(\mathrm{Re}(\tilde{\mathsf r}_i^c) - \mathrm{Re}(\mathsf x) -\left[\frac{\mathrm{Re}(\tilde{\mathsf r}_i^c) - \mathrm{Re}(\mathsf x)}{M} \right] \right)^2 \right.
\left.+\left(\mathrm{Im}(\tilde{\mathsf r}_i^c)-\mathrm{Im}(\mathsf x)
-\left[\frac{\mathrm{Im}(\tilde{\mathsf r}_i^c)-\mathrm{Im}(\mathsf x)}{M} \right] \right)^{2}\right)\\
&  = \sum_{i=1}^{L}w_id_M^2\left(\mathrm{Re}(\tilde{\mathsf r}_i^c), \mathrm{Re}(\mathsf x)\right) +\sum_{i=1}^{L}w_id_M^2\left(\mathrm{Im}(\tilde{\mathsf r}_i^c), \mathrm{Im}(\mathsf x)\right).\\
\end{aligned}$$
Thus, $\sum_{i=1}^Lw_i|d_M\left(\tilde{\mathsf r}_i^c, \mathsf x\right)|^{2}$ attains its minimum value if and only if both
$\sum_{i=1}^{L}w_id_M^2\left(\mathrm{Re}(\tilde{\mathsf r}_i^c), \mathrm{Re}(\mathsf x)\right)$ and
$\sum_{i=1}^{L}w_id_M^2\left(\mathrm{Im}(\tilde{\mathsf r}_i^c), \mathrm{Im}(\mathsf x)\right)$ attain their minimum values, since region $\mathcal F_M$ of variable $\mathsf x$ is a square with sides parallel to the two axes and thus $\mathrm{Re}(\mathsf x)$ and $\mathrm{Im}(\mathsf x)$ are independent each other.
\end{proof}

For real numbers, the optimal estimate $\hat r^c$ is provided in Theorem 2 of \cite{Wang_2015}. For complex numbers, we obtain the following result.

\begin{theorem}\label{Th_Omega_r_c}
The optimal estimate $\hat{\mathsf r}^c$ in (\ref{estimate_rc_1}) belongs to the following set:
\begin{align}
 \Omega = \left\{\left\langle
  \sum_{i=1}^Lw_i\tilde{\mathsf r}_i^c +M\left(\sum_{i=1}^{k_1}w_{\varsigma_{(i)}} + \mathrm{i}\sum_{i=1}^{k_2}w_{\upsilon_{(i)}}\right)\right\rangle_M:
  k_1, k_2=1, 2, \ldots, L\right\},
\end{align}
where $\varsigma$ and $\upsilon$ are permutations on $\{1, 2, \ldots, L\}$ satisfying
$$\mathrm{Re}(\tilde{\mathsf r}_{\varsigma_{(1)}}^c)\leq\mathrm{Re}(\tilde{\mathsf r}_{\varsigma_{(2)}}^c) \leq \cdots \leq \mathrm{Re}(\tilde{\mathsf r}_{\varsigma_{(L)}}^c)$$
and
$$\mathrm{Im}(\tilde{\mathsf r}_{\upsilon_{(1)}}^c)
\leq\mathrm{Im}(\tilde{\mathsf r}_{\upsilon_{(2)}}^c)
\leq\cdots\leq\mathrm{Im}(\tilde{\mathsf r}_{\upsilon_{(L)}}^c),$$
respectively.
\end{theorem}

\begin{proof}
Based on Theorem \ref{Hat_rcImage}, it suffices to solve for $\mathrm{Re}(\hat{\mathsf r}^c)$ and $\mathrm{Im}(\hat{\mathsf r}^c)$ in (\ref{Re_rc_Im_rc}). By Theorem 2 of \cite{Wang_2015}, we have
\begin{equation}\label{Omega_Re_r}
\mathrm{Re}(\hat{\mathsf r}^c) = \left\langle\sum_{i=1}^L w_i\mathrm{Re}(\tilde{\mathsf r}_i^c) + M\sum_{i=1}^{k_1}w_{\varsigma_{(i)}} \right\rangle_M
\end{equation}
and
\begin{equation}\label{Omega_Im_r}
\mathrm{Im}(\hat{\mathsf r}^c) = \left\langle\sum_{i=1}^L w_i\mathrm{Im}(\tilde {\mathsf r}_i^c) + M\sum_{i=1}^{k_2}w_{\upsilon_{(i)}} \right\rangle_M
\end{equation}
for some integers $k_1, k_2\in \{1, 2, \ldots, L\}$. Hence,
$$ \hat{\mathsf r}^c = \left\langle\sum_{i=1}^L w_i\tilde{\mathsf r}_i^c + M\left(\sum_{i=1}^{k_1}w_{\varsigma_{(i)}} + \mathrm{i}\sum_{i=1}^{k_2}w_{\upsilon   _{(i)}}\right)\right\rangle_M.$$
This completes the proof of the theorem.
\end{proof}

In terms of the set $\Omega$ in Theorem \ref{Th_Omega_r_c}, there are $L^2$ possible candidates for obtaining the optimal estimate. The following result demonstrates that the number of searches can be reduced from $L^2$ to $2L$.

\begin{corollary}
The optimal estimate $\hat{\mathsf r}^c$ can be determined with a total of $2L$ searches.
\end{corollary}

\begin{proof}
By (\ref{Omega_Re_r}) and (\ref{Omega_Im_r}) in the proof of Theorem \ref{Th_Omega_r_c}, both $\mathrm{Re}(\hat{\mathsf r}^c)$ and $\mathrm{Im}(\hat{\mathsf r}^c)$ can be obtained through $L$ searches respectively.
Thus, $\hat{\mathsf r}^c$ can be determined with $2L$ searches.
\end{proof}

\textbf{Comparison with the Two-Stage CRT in \cite{Gong_2021}}:
The two-stage CRT described in \cite{Gong_2021} consists of two stages, each applying the closed-form CRT \cite{Wang_2010}. In the first stage, $l$ pairs of equations with complex moduli are converted into $l$ equations with real moduli, where $2l \le L$. In the second stage, the real and imaginary parts are separated, and two congruence systems are solved using the closed-form CRT. The reconstruction result depends on the choice of the reference remainder, which is influenced by the estimation of the common remainder.
Note that the estimation of the common remainder is based on the searching method in \cite{Wang_2010} used in \cite{Gong_2021}, although the fast searching of only $L$ times in \cite{Wang_2015} can be applied. This requires searching through all the points within the interval $[0, M)$ and the estimation depends on the searching step sizes. To achieve good estimation accuracy, small searching step sizes are required, resulting in many more searching steps than $2L$. Notably, the complexity of the two-stage CRT increases as $M$ increases.
Furthermore, it is based on the assumption that the remainder errors have the same variance. If the variances differ, the reconstruction performance is significantly degraded.

\section{Robust Estimation for the Fast MLE C-CRT}\label{sec4}
In this section, we present a necessary and sufficient condition for the MLE C-CRT to be robust. Then, we calculate the probability of the robust MLE C-CRT.

\subsection{Condition of Robust Estimation}\label{sec4_A}
We first consider a necessary condition of robust estimation for the MLE C-CRT. For convenience, we define the remainder error set as
$$\mathcal U = \left\{\Delta \mathsf r_1, \Delta \mathsf r_2, \ldots, \Delta \mathsf r_L\right\}$$
and the weighted average of the remainder errors as
$$\overline{\Delta \mathsf r} = \sum_{i=1}^{L}w_i \Delta \mathsf r_i,$$
where the weights $w_i$ are defined in (\ref{eq_weights}). In \cite{Wang_2015}, a necessary condition for a robust estimation of real numbers is
$$-\frac{M}{2} \le \Delta r_i - (\hat N - N) < \frac{M}{2}.$$
For complex numbers, we have the following necessary condition:
\begin{eqnarray}\label{condition_delta}
\Delta \mathsf r_i - (\hat{\mathsf N} - \mathsf N) \in \mathcal S_M,
\end{eqnarray}
where $\mathcal S_M$ is defined in (\ref{eq_def_SM}). In what follows, in order to discuss the robustness of the MLE C-CRT, we suppose that (\ref{condition_delta}) is always satisfied.
%
%\begin{lemma}\label{ThmReals}
%\cite{Wang_2015} If the weighted average error for real numbers satisfy $\left|\overline {\Delta r}\right| < \frac{M}{2}$ and
%\begin{equation}
%-\frac{M}{2}\le \sum_{\Delta r_i\in \mathcal V} \frac{w_i\Delta r_i}{\sum_{\Delta r_j\in \mathcal V} w_j} - \sum_{\Delta r_i\in \overline{\mathcal V}} \frac {w_i\Delta r_i}{\sum_{\Delta r_j\in \overline{\mathcal V}} w_j} < \frac{M}{2}
%\end{equation}
%for any subset $\mathcal V$ of $\mathcal U$, where $\overline{\mathcal V}=\mathcal U\setminus \mathcal V$ is the complement of $\mathcal V$ in $\mathcal U$, then the optimal $\hat{r}^c$ for real numbers has the form
%\begin{equation*}
%\hat r^c  = \left\langle r^c + \overline {\Delta r}\right\rangle_M.
%\end{equation*}
%\end{lemma}

\begin{theorem}\label{thm_Delta_r}
Suppose that $\overline{\Delta \mathsf r}$ satisfies $|\mathrm{Re}(\overline{\Delta \mathsf r})| < \frac{M}{2}$ and $|\mathrm{Im}(\overline{\Delta \mathsf r})| < \frac{M}{2}$ simultaneously. If
\begin{equation}\label{robust_ri_SM}
\sum_{\Delta \mathsf r_i\in \mathcal V}\frac{w_i\Delta \mathsf r_i}{\sum_{\Delta \mathsf r_j\in \mathcal V}w_j} - \sum_{\Delta \mathsf r_i\in\overline{\mathcal V}}\frac{w_i\Delta \mathsf r_i}{\sum_{\Delta \mathsf r_j \in\overline{\mathcal V}}w_j}\in \mathcal S_M
\end{equation}
holds for any $\mathcal V\subseteq\mathcal U$ and $\overline{\mathcal V} =\mathcal U \backslash \mathcal V$, then we have
\begin{eqnarray}\label{Tm4Hat_rc}
\hat{\mathsf r}^c = \langle \mathsf r^c + \overline{\Delta \mathsf r}\rangle_M.
\end{eqnarray}
Furthermore,
\begin{equation}\label{Re drc}\hspace{-0.95em}
\mathrm{Re}(\Delta \mathsf r^c) = \mathrm{Re}(\overline{\Delta\mathsf r})+
\begin{cases}
M, &\text{if}~ \mathrm{Re}(\mathsf r^c+\overline{\Delta \mathsf r}) < 0, \\
0, &\text{if}~ 0\leq \mathrm{Re}(\mathsf r^c+\overline{\Delta \mathsf r}) < M,\\
-M, &\text{if}~ \mathrm{Re}(\mathsf r^c+\overline{\Delta \mathsf r}) \ge M
\end{cases}
\end{equation}
and
\begin{equation}\label{Im drc}\hspace{-0.98em}
\mathrm{Im}(\Delta\mathsf r^c) = \mathrm{Im}(\overline{\Delta \mathsf r})+
\begin{cases}
M, &\text{if}~ \mathrm{Im}(\mathsf r^c+\overline{\Delta \mathsf r}) < 0, \\
0, &\text{if}~ 0\leq \mathrm{Im}(\mathsf r^c+\overline{\Delta \mathsf r}) < M, \\
-M, &\text{if}~ \mathrm{Im}(\mathsf r^c+\overline{\Delta \mathsf r})\ge M.
\end{cases}
\end{equation}
\end{theorem}

The proof of this theorem is in Appendix $B$.

Theorem \ref{thm_Delta_r} gives a condition of the remainder errors and their weights such that the optimal estimate $\hat{\mathsf r}^c$ of the common remainder $\mathsf r^c$ is $\langle \mathsf r^c + \overline {\Delta \mathsf r}\rangle_M$. Next, we consider the sufficiency of the robust estimation when the errors satisfy (\ref{robust_ri_SM}). For convenience, we introduce a result below.

\begin{proposition}\label{Th_Set_N_N_0}
For $\mathsf N$ and $\mathsf N_0$ in (\ref{eq_N_r_c}),
$\mathsf N_0\in \mathcal H$ if and only if $\mathsf N \in M\mathcal H =\{M\mathsf h: \mathsf h\in\mathcal H\}$, where
$$\mathcal H=\{h_1 + h_2 \mathrm i: 1 \leq h_1, h_2 < \Gamma-1\}.$$
\end{proposition}

\begin{proof}
According to (\ref{eq_q_iNN0}), it suffices to prove that $\mathsf N-\mathsf r^c\in M\mathcal H$ if and only if
$\mathsf N \in M\mathcal H$.
By $\mathsf N - \mathsf r^c=M \left\lfloor\frac{\mathsf N}{M}\right\rfloor$,
we can obtain that $\mathsf N-\mathsf r^c\in M\mathcal H$ if and only if $\left\lfloor\frac{\mathsf N}{M}\right\rfloor\in\mathcal H$,
which is equivalent to $\mathsf N \in M\mathcal H$.
\end{proof}

According to Theorem \ref{thm_Delta_r}, we have
$\Delta\mathsf r^c = \overline{\Delta\mathsf r} + Mk_1 + Mk_2\mathrm i$,
where $k_1, k_2\in\{-1, 0, 1 \}$. By the definition of $\hat{\mathsf q}_i$ in (\ref{eq:estimating_q_i}), we have
\begin{equation}\label{temp1_Thm4}
\hat{\mathsf q}_i = \mathsf q_i + \left[\frac{\Delta\mathsf r_i - \Delta\mathsf r^c}{M}\right] = \mathsf q_i - k_1 - k_2\mathrm i + \left[\frac{\Delta\mathsf r_i - \overline{\Delta \mathsf r}}{M}\right].
\end{equation}
Arbitrarily choose a $\Delta\mathsf r_i\in\mathcal U$, and let $\mathcal V = \{\Delta\mathsf r_i\}$.
Then, we obtain from (\ref{robust_ri_SM}) that
$$\Delta\mathsf r_i - \sum_{j\neq i}\frac{w_j}{\sum_{j\neq i}w_j}\Delta\mathsf r_j \in \mathcal S_M.$$
Since $\sum_{j\neq i}w_j = 1 - w_i$, we have
$$\frac{1}{1 - w_i}\Delta\mathsf r_i - \frac{1}{1-w_i}\overline{\Delta\mathsf r} \in \mathcal S_M.$$
Hence, $$\Delta\mathsf r_i - \overline{\Delta\mathsf r} \in (1-w_i)\mathcal S_M.$$
Consequently,
$$\left[\frac{\Delta\mathsf r_i - \overline{\Delta\mathsf r}}{M}\right] = 0.$$
It follows from (\ref{temp1_Thm4}) that
$\hat{\mathsf q}_i = \mathsf q_i - k_1 -k_2 \mathrm i$.
Based on the arbitrariness of $\Delta\mathsf r_i$, we can obtain from (\ref{eq:calculate_N_0}) that
\begin{equation}\label{eq_hat_N_0}
\hat{\mathsf N}_0 \equiv \sum_{i=1}^L\bar{\mathsf\gamma}_i\mathsf \gamma_i\mathsf q_i - \sum_{i=1}^L\bar{\mathsf\gamma}_i
\mathsf \gamma_i(k_1 + k_2\mathrm i) \mod \Gamma.
\end{equation}
Note that $\bar{\mathsf\gamma}_i \mathsf \gamma_i \equiv 1 \mod \mathsf \Gamma_i$ and $\bar{\mathsf\gamma}_j \mathsf \gamma_j \equiv 0 \mod \mathsf \Gamma_i$ for $j \ne i$. Hence, $$\sum_{i=1}^L\bar{\mathsf\gamma}_i
\mathsf \gamma_i \equiv 1\mod \mathsf \Gamma_i.$$
Consequently, $$\sum_{i=1}^L\bar{\mathsf\gamma}_i
\mathsf \gamma_i (k_1 + k_2\mathrm i) \equiv k_1 + k_2\mathrm i \mod \mathsf \Gamma_i,$$ that is, $\mathsf \Gamma_i$ divides $\sum_{i=1}^L\bar{\mathsf\gamma}_i
\mathsf \gamma_i (k_1 + k_2\mathrm i) - (k_1 + k_2\mathrm i)$. Since $\mathsf\Gamma_1, \mathsf\Gamma_2, \ldots, \mathsf\Gamma_L$ are pairwise coprime, $\prod_{i=1}^L\mathsf\Gamma_i$ divides $\sum_{i=1}^L\bar{\mathsf\gamma}_i
\mathsf \gamma_i (k_1 + k_2\mathrm i) - (k_1 + k_2\mathrm i)$. Thus,
$$\sum_{i=1}^L\bar{\mathsf\gamma}_i \mathsf \gamma_i(k_1 + k_2\mathrm i)\equiv k_1+k_2\mathrm i\mod \Gamma.$$ By (\ref{eq_hat_N_0}), we have
$$\hat{\mathsf N}_0\equiv \mathsf N_0-k_1-k_2\mathrm i\mod \Gamma.$$
Since $\hat{\mathsf N}_0\in\mathcal F_{\Gamma}$, we have
$$\hat{\mathsf N}_0 = \left\langle\mathsf N_0 - k_1 - k_2 \mathrm i\right\rangle_\Gamma.$$
If $\mathsf N \in M\mathcal H$, i.e., $\mathsf N_0\in\mathcal H$ by Proposition \ref{Th_Set_N_N_0}, then $\hat{\mathsf N}_0 = \mathsf N_0 - k_1 -k_2 \mathrm i$. Consequently,
$$\hat{\mathsf N} = M(\mathsf N_0 - k_1 -k_2\mathrm i) + \mathsf r^c + \Delta \mathsf r^c = \mathsf N + \overline{\Delta\mathsf r}.$$
Therefore, $\mathsf N$ can be robustly estimated. The next theorem demonstrates that (\ref{robust_ri_SM}) is both a necessary and sufficient condition for robust estimation.

\begin{theorem}\label{Theo_iff}
Let $\mathsf N \in M\mathcal H$. If $\overline {\Delta \mathsf r}$ satisfies $|\mathrm{Re}(\overline{\Delta \mathsf r})| < \frac{M}{2}$ and $|\mathrm{Im}(\overline{\Delta \mathsf r})| < \frac{M}{2}$ simultaneously, then
\begin{equation}\label{hat_N_N}
\hat{\mathsf N} - \mathsf N = \overline{\Delta \mathsf r}
\end{equation}
holds if and only if (\ref{robust_ri_SM}) holds for all $\mathcal V\subseteq\mathcal U$.
\end{theorem}

The proof of this theorem can be found in Appendix $C$.

Note that $M\mathcal H=\left\{M\Gamma(a+b\mathrm i): \frac{1}{\Gamma}\leq a, b<1-\frac{1}{\Gamma}\right\}$.
As illustrated in Fig. \ref{N_0_N}, compared to $\mathcal F_{M\Gamma}$, $M\mathcal H$ only differs with (does not include) four trapezoids with a height of $M$.
\begin{figure}[h!]
\vspace{-0.3cm}
\centering
\captionsetup{skip=-18pt}
\setlength{\belowcaptionskip}{-0.3cm}
\captionsetup{font={footnotesize}}
\captionsetup{labelsep=period}
\includegraphics[width=0.3\textwidth]{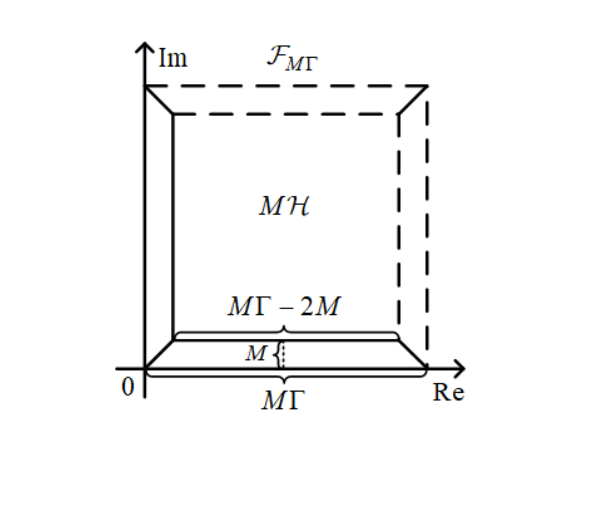}
\caption{Illustration of $\mathcal F_{M\Gamma}$ and $M\mathcal H$.}\label{N_0_N}
\end{figure}

Theorem \ref{Theo_iff} demonstrates that the MLE C-CRT is capable of ``preserving errors'', i.e., it preserves the weighted average error $\overline{\Delta \mathsf r}$ of the original remainder errors $\Delta \mathsf r_1, \Delta \mathsf r_2, \ldots, \Delta \mathsf r_L$ during the reconstruction. In this case, it is called error preserving MLE C-CRT. Since the output error of the MLE C-CRT is comparable to the input remainder error level, it is robust. By Theorem \ref{Theo_iff}, we can derive the following robustness as well.

\begin{corollary}\label{Cor_lb_zt}
Let $\mathsf N \in M\mathcal H$, $|\mathrm{Re}(\overline{\Delta \mathsf r})| < \tau$ and $|\mathrm{Im}(\overline{\Delta \mathsf r})| < \tau$, where $\tau\leq \frac{M}{2}$.
If (\ref{robust_ri_SM}) holds for all $\mathcal V\subseteq\mathcal U$, then $|\mathrm{Re}(\hat{\mathsf N}) - \mathrm{Re}(\mathsf N)|< \tau$ and
$|\mathrm{Im}(\hat{\mathsf N}) - \mathrm{Im}(\mathsf N)|< \tau$.
\end{corollary}

\begin{proof}
Clearly, $|\mathrm{Re}(\overline{\Delta \mathsf r})| < \frac{M}{2}$ and $|\mathrm{Im}(\overline{\Delta \mathsf r})| < \frac{M}{2}$. By Theorem \ref{Theo_iff}, we have $\hat{\mathsf N} - \mathsf N = \overline{\Delta \mathsf r}$. This leads to $|\mathrm{Re}(\hat{\mathsf N}) - \mathrm{Re}(\mathsf N)|< \tau$ and $|\mathrm{Im}(\hat{\mathsf N}) - \mathrm{Im}(\mathsf N)|< \tau$.
\end{proof}

Corollary \ref{Cor_lb_zt} presents the robustness of MLE C-CRT in terms of a bound of the weighted average error $\overline{\Delta \mathsf r}$. The following result presents the conditions under which all remainder errors contribute to a robust estimation, which is analogous to the error bound $\frac{M}{4}$ for real numbers \cite{Li_2009, Wang_2010}.

\begin{corollary}\label{Omega_err}
Let $\mathsf N \in M\mathcal H$. If $|\mathrm{Re}(\Delta \mathsf r_i)| <\tau$ and $|\mathrm{Im}(\Delta \mathsf r_i)| <\tau$ hold for all $i=1, 2, \ldots, L$,
where $\tau\leq\frac M4$, then we have
$|\mathrm{Re}(\hat{\mathsf N}) - \mathrm{Re}(\mathsf N)|
<\tau$ and $|\mathrm{Im}(\hat{\mathsf N}) - \mathrm{Im}(\mathsf N)|<\tau$.
\end{corollary}

\begin{proof}
Since $\mathrm{Re}(\Delta \mathsf r_{i})<\tau$ and $\mathrm{Im}(\Delta \mathsf r_{i})<\tau$, we have
$$\left|\sum_{\Delta r_{i}\in S}\frac{w_i\mathrm{Re}(\Delta \mathsf r_{i})}{\sum_{\Delta r_{j}\in S}w_{j}}-\sum_{\Delta r_{i}\in\overline{S}}\frac{w_i\mathrm{Re}(\Delta \mathsf r_{i})}{\sum_{\Delta r_{j}\in\overline{S}}w_{j}}\right| < \left|\sum_{\Delta r_{i}\in S}\frac{w_i\tau}{\sum_{\Delta r_{j}\in S}w_{j}}\right| + \left|\sum_{\Delta r_{i}\in\overline{S}}\frac{w_i\tau}{\sum_{\Delta r_{j}\in\overline{S}}w_{j}}\right| =2\tau.$$
Similarly,
$$\left|\sum_{\Delta r_{i}\in S}\frac{w_i\mathrm{Im}(\Delta \mathsf r_{i})}{\sum_{\Delta r_{j}\in S}w_{j}}-\sum_{\Delta r_{i}\in\overline{S}}\frac{w_i\mathrm{Im}(\Delta \mathsf r_{i})}{\sum_{\Delta r_{j}\in\overline{S}}w_{j}}\right|< 2\tau.$$
By Theorem \ref{Theo_iff}, we have
$|\mathrm{Re}(\hat{\mathsf N}) - \mathrm{Re}(\mathsf N)| = |\mathrm{Re}(\overline{\Delta \mathsf r})|<\tau$ and $|\mathrm{Im}(\hat{\mathsf N}) - \mathrm{Im}(\mathsf N)|= |\mathrm{Im}(\overline{\Delta \mathsf r})|<\tau$.
\end{proof}
This result gives a concrete robust 2D-CRT compared to the general setting in \cite{Xiao_2020} and \cite{Xiao_2024}.

\subsection{Probability of Error Preserving MLE C-CRT}\label{sec4_B}

We now calculate the probability of the MLE C-CRT preserving errors, i.e., satisfying the necessary and sufficient conditions in Theorem \ref{Theo_iff}. This is also a probability for achieving robust reconstruction when $\tau=\frac{M}{2}$ by Corollary \ref{Cor_lb_zt}. Assume that the $i$-th remainder error $\Delta \mathsf r_i$ follows a wrapped complex Gaussian distribution with a mean of $0$ and a variance of $2\sigma_i^2$, and that the real and imaginary parts of $\Delta \mathsf r_i$ have equal variance $\sigma_i^2$. Since $|M\Gamma_i|$ is generally much larger than $\sigma_i^2$, we approximate $\Delta\mathsf r_i$ as a complex Gaussian distribution.

According to Theorem \ref{Theo_iff}, the necessary and sufficient condition for the MLE C-CRT to robustly estimate $\mathsf N$ is that the errors $\Delta\mathsf r_i$ satisfy $(\ref{robust_ri_SM})$.
Denote $x_i = \mathrm{Re}(\Delta\mathsf r_i)$ and $y_i = \mathrm{Im}(\Delta \mathsf r_i)$.
Let $\mathcal R$ be the set of all vectors $\mathbf r = (\Delta \mathsf r_1, \Delta\mathsf r_2, \ldots, \Delta\mathsf r_L)$ that satisfy $(\ref{robust_ri_SM})$. Similar to the discussion for real numbers in \cite{Wang_2015}, we have
$$\hspace{-0.3em}\begin{aligned}
& p((\Delta\mathsf r_1, \Delta\mathsf r_2, \ldots, \Delta\mathsf r_L)\in\mathcal R) \nonumber\\
= &  \frac{1}{(2\pi)^{L}}\prod_{i=1}^{L} \frac{1}{\sigma_i^{2}}
\underset{\mathbf r \in \mathcal R}{\idotsint } \exp{\left\{\sum_{i=1}^L\left(-\frac{x_i^2}{2\sigma_i^2} - \frac{y_i^2}{2\sigma_i^2}\right)\right\}} \mathrm{d}V_\mathbf x \mathrm{d}V_\mathbf y \nonumber \\
= &  \frac{1}{(2\pi)^{L}}\prod_{i = 1}^{L} \frac{1}{\sigma_i^{2}}
\bigg(\underset{\mathbf r \in \mathcal R}{\idotsint }
\exp\left\{-\sum_{i = 1}^L \frac{x_i^2}{2\sigma_i^2}\right\}\mathrm{d}V_\mathbf x\bigg)^2,
\end{aligned}$$
where $\mathbf x = (x_1, x_2, \ldots, x_L)$ and $\mathbf y = (y_1, y_2, \ldots, y_L)$, $\mathrm dV_\mathbf x$ and $\mathrm dV_\mathbf y$ are the differential volume elements of $\mathbf x$ and $\mathbf y$, respectively.

Fig. \ref{Probability_L_2} illustrates the theoretical and simulated probabilities of the error preserving MLE C-CRT at different standard deviations of the real and imaginary parts of noises. In the simulation, we set $L=2$, $M=10$, $\mathsf\Gamma_1= 4+19\mathrm i$, $\mathsf\Gamma_2 = 4-19\mathrm i$, $\sigma_1=2.4 + k$, and $\sigma_2 =2.5 + k$. For each $k$, the number of trials is $100000$.

\begin{figure}[htbp]
\centering
\setlength{\belowcaptionskip}{-0.2cm}
\captionsetup{font={footnotesize}}
\captionsetup{labelsep=period}
\includegraphics[width=0.5\textwidth]{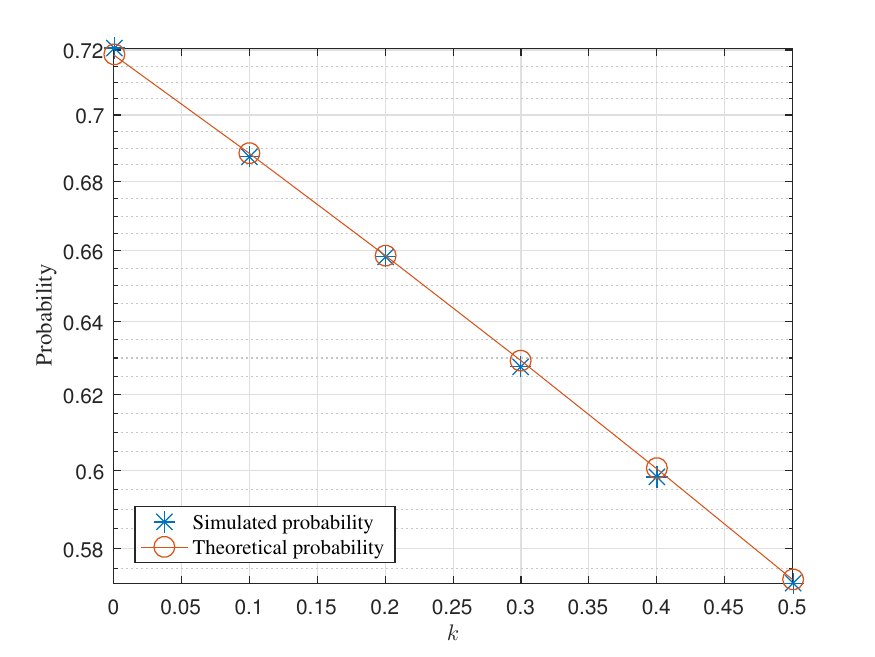}
\caption{Probabilities of error preserving MLE C-CRT (i.e., robust) at different standard deviations.}\label{Probability_L_2}
\end{figure}

\section{Simulation Results}\label{sec5}
In this section, we present several simulations to demonstrate the effectiveness of the proposed fast MLE C-CRT algorithm. Furthermore, we apply the algorithm to modulo ADCs to highlight its practical applicability.

We first present the relationship between signal-to-noise ratio (SNR) and noise variance. Let a complex-valued bandlimited signal be $\tilde{\mathsf g}(t)=\mathsf g(t)+\mathsf w(t)$, where $\mathsf w(t)$ is the noise with mean 0 and variance $2\sigma^2$. Then, the SNR is defined as
$$ \mathsf g_\mathrm{SNR} = 10\log_{10}\frac{\sum_n|\mathsf g(n/f_s)|^2}{\sum_n|\mathsf w(n/f_s)|^2},$$
where $f_s$ is the sampling frequency. Assume that the sampled values of $\mathsf g(t)$ are uniformly distributed within $\mathcal F_{\mathsf M}$, where $\mathsf M\in\mathbb Z[\mathrm i]$.
When there are enough sampled values, the mean values of $|\mathsf g(n/f_s)|^2$ and $|\mathsf w(n/f_s)|^2$ can be approximated as $\frac{2}{3}|\mathsf M|^2$ and $2\sigma^2$, respectively. Hence,
\begin{equation}\label{eq_SNR_sigma}
\mathsf g_\mathrm{SNR} \approx 10\log_{10}\frac{|\mathsf M|^2}{3\sigma^2}.
\end{equation}

\subsection{Comparison of Fast MLE C-CRT and Two-Stage CRT Algorithms}
In this subsection we compare the proposed fast MLE C-CRT and the two-stage CRT
presented in \cite{Gong_2021} in terms of both performance and computational complexity, where the moduli $\mathsf \Gamma_1, \mathsf\Gamma_2, \ldots, \mathsf\Gamma_8$ are set as $1+4\mathrm{i}, 1-4\mathrm{i}, 3+4\mathrm{i}, 3-4\mathrm{i}, 2+7\mathrm{i}, 2-7\mathrm{i}, 3, 7$, respectively.
In each trial, the real and imaginary parts of the complex number $\mathsf N$ are randomly selected from the interval $[M, M(\Gamma-1))$, where $M = 10$. The real and imaginary parts of the remainder errors $\Delta \mathsf r_i$
follow a wrapped complex Gaussian distribution with mean $0$ and variance $\sigma_i^2$. In the simulation, we set $\sigma_i = u|M\mathsf\Gamma_i|$, where $u$ is a small positive constant. For convenience, we approximate $\mathsf r_i$ as a uniform distribution within $\mathcal F_{M\mathsf\Gamma_i}$. Similar to (\ref{eq_SNR_sigma}), $-20\log_{10}\sqrt{3}u$ can be used as the measurement for SNR of the remainders. For each $u$, the total number $n$ of trials is $10000$, i.e., $n=10000$. We evaluate the performance of the two methods using two metrics: the root mean square error (RMSE), and the trial fail rate (TFR) for robust reconstruction and preserving errors.
The RMSE of $\mathsf N$ is defined as
$$ \Delta \mathsf N_{\text{RMSE}} = \sqrt{\frac{1}{n}\sum_{j=1}^n|\mathsf N_j - \hat{\mathsf N}_j|^2}.$$
According to Theorem \ref{Theo_iff}, the theoretical RMSE for the fast MLE C-CRT is
$$ \Delta\mathsf N_{\text{theory}} = \sqrt{E\left\{\left(\mathrm{Re}(\overline{\Delta \mathsf r})\right)^2\right\}+E\left\{\left(\mathrm{Im}(\overline{\Delta \mathsf r})\right)^2\right\}},$$
where $E\{\cdot\}$ denotes the mean. Since $\mathrm{Re}(\Delta \mathsf r_i)$ are mutually independent and Gaussian distributed for $i=1, 2, \ldots, L$, $\mathrm{Re}(\overline{\Delta \mathsf r})$ follows a Gaussian distribution with mean $0$ and variance $\sum_{i=1}^L w_i^2\sigma_i^2$. Similarly, the distribution of $\mathrm{Im}(\overline{\Delta \mathsf r})$ is the same as that of $\mathrm{Re}(\overline{\Delta \mathsf r})$. Thus,
\begin{equation}\label{eq_N_throry}
\Delta\mathsf N_{\text{theory}}=\sqrt{2\sum_{i=1}^L w_i^2\sigma_i^2}.
\end{equation}

\begin{figure}[htbp]
\centering
\setlength{\belowcaptionskip}{-0.2cm}
\captionsetup{font={footnotesize}}
\captionsetup{labelsep=period}
\includegraphics[width=0.5\textwidth]{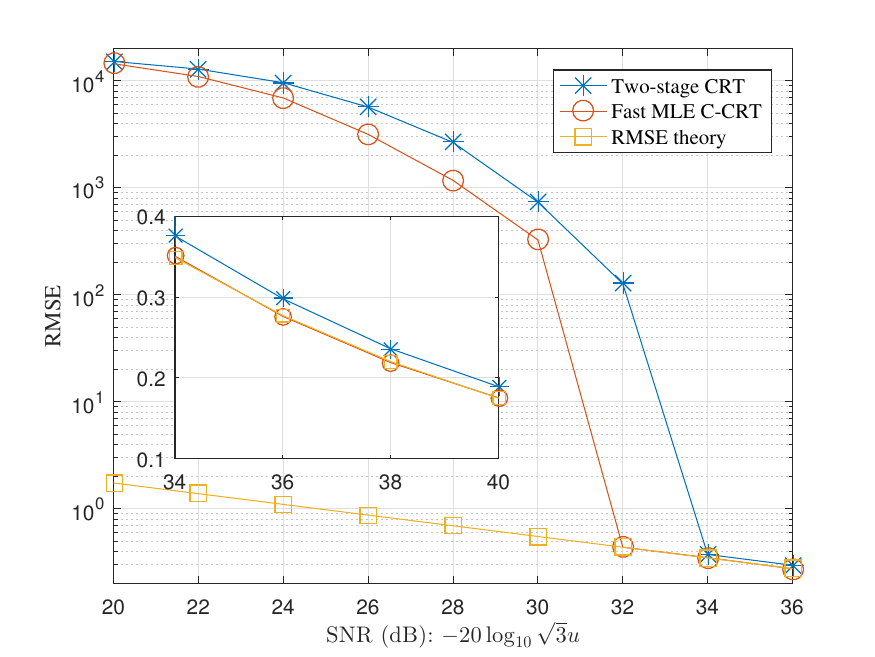}
\caption{Comparison of the RMSE.}\label{RMSE_duibi}
\end{figure}

Fig. \ref{RMSE_duibi} illustrates the curves of the RMSE for the two algorithms in terms of SNR, along with the theoretical RMSE of the fast MLE C-CRT. It can be observed that the RMSE of the fast MLE C-CRT is smaller than that of the two-stage CRT. At the SNR of $32$dB, the maximum error in the real and imaginary parts
of the remainders is $4.3267$, the reconstruction errors for the fast MLE C-CRT are less than $1.2621$. At the SNR of $34$dB, the maximum error in the remainders is $3.1640$, the reconstruction errors for the two-stage CRT are less than $1.1697$.
From Fig. \ref{RMSE_duibi}, one can see that the fast MLE C-CRT achieves robust reconstruction more quickly than the two-stage CRT does, since it has the optimal estimation of the common remainder. When the SNR is less than $30$dB, due to the errors of some remainders exceeding $\frac{M}{2}$, this does not satisfy the conditions of Theorem \ref{Theo_iff}. Hence, the fast MLE C-CRT may not have robust reconstruction and has large errors. Similarly, when the SNR is less than $34$dB, the condition for the robust reconstruction of the two-stage CRT algorithm may not be satisfied and thus has large errors. On the other hand, the theoretical curve is from (\ref{eq_N_throry}) that is based on the assumption of complex Gaussian distributions of the remainder errors and only depends on the error distribution variances, while the true distributions of the remainder errors follow wrapped complex Gaussian distributions. It justifies that the theoretical curve is smooth and does not exhibit a large change. When the SNR is higher, the assumption holds better and the simulated curve and the theoretical curve match better. As one can see from the zoom-in part in Fig. \ref{RMSE_duibi}, the fast MLE C-CRT, in fact, achieves the theoretical RMSE values when SNR is high, while the two-stage CRT cannot.

For the TFR of the robust reconstruction, we consider the estimation error of $\mathsf N_j$. If $\mathsf N_j$ and $\hat{\mathsf N}_j$ satisfy $|\mathrm{Re}(\mathsf N_j - \hat{\mathsf N}_j)| < \tau$
and $|\mathrm{Im}(\mathsf N_j - \hat{\mathsf N}_j)| < \tau$
simultaneously for a pre-given small positive constant $\tau$, the trial is considered successful and otherwise, the trial fails. In the simulations, we set $\tau = \frac{M}{4} = 2.5$, which is the upper bound of the errors in Corollary \ref{Omega_err}. Fig. \ref{tau_duibi} illustrates the curves of the TFR for the two algorithms in terms of SNR. It is evident that the fast MLE C-CRT outperforms the two-stage CRT, particularly at higher SNR values.

\begin{figure}[htbp]
\centering
\setlength{\belowcaptionskip}{-0.2cm}
\captionsetup{font={footnotesize}}
\captionsetup{labelsep=period}
\includegraphics[width=0.5\textwidth]{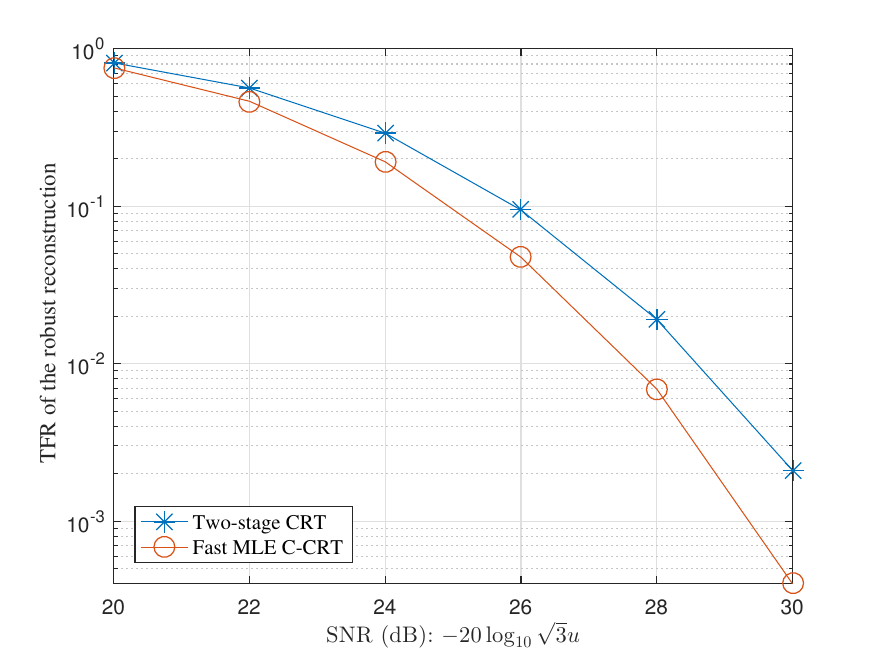}
\caption{Comparison of TFR.}
\label{tau_duibi}
\end{figure}

For the TFR in terms of preserving errors, we test whether (\ref{hat_N_N}) holds. If it does, the trial is considered successful and otherwise, the trial fails. According to Theorem \ref{Theo_iff}, the TFR of preserving errors for each trial can be expressed as
$$p_{\text{TFR}}^L = 1 - p((\Delta\mathsf r_1, \Delta\mathsf r_2, \ldots, \Delta\mathsf r_L)\in\mathcal R).$$
Fig. \ref{TRT_tong} illustrates the curves of the TFR and its theoretical value for the fast MLE C-CRT with respect to the number of moduli $L$, where the noise variances $2\sigma_i^2$ are constant and equal to $2\sigma^2$ for $i=1, 2, \ldots, L$. It can be observed that the TFR of the fast MLE C-CRT closely matches its theoretical value in both cases.

\begin{figure}[htbp]
\centering
\setlength{\belowcaptionskip}{-0.2cm}
\captionsetup{font={footnotesize}}
\captionsetup{labelsep=period}
\includegraphics[width=0.45\textwidth]{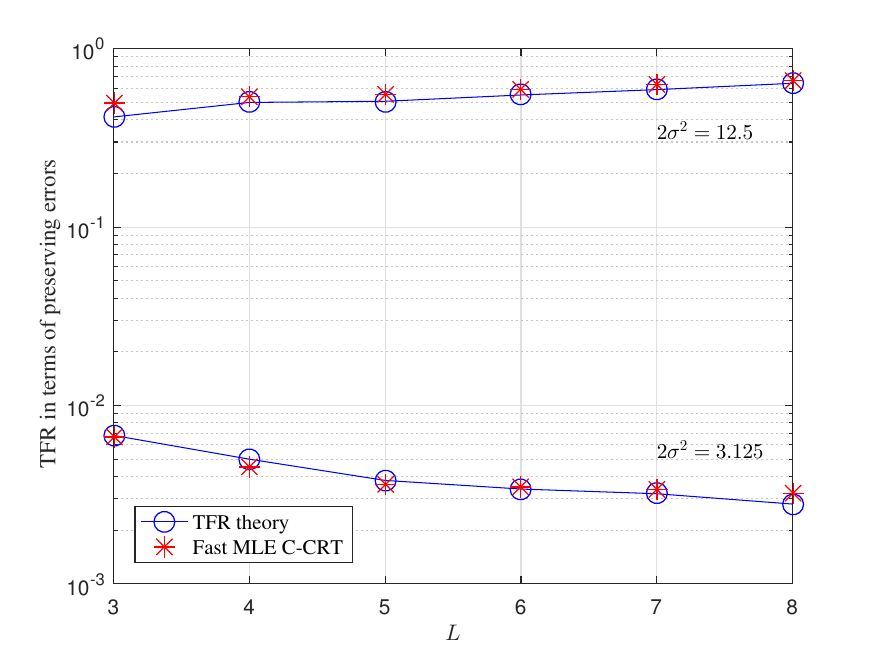}
\caption{TFR for different noise variances.}\label{TRT_tong}
\end{figure}

We now compare the computational complexities of the two methods by counting the numbers of the real multiplications they need, i.e., the real multiplicative complexities. Given a known common remainder, the real multiplicative complexity for the fast MLE C-CRT can be directly calculated as $\mathcal O(L)$ by (\ref{eq:estimating_q_i}), (\ref{eq:calculate_N_0}), and (\ref{eq:calculate_N2B}).
The two-stage CRT requires twice the number of operations of the closed-form CRT for reals, with a real multiplicative complexity of $\mathcal O(L)$ when the common remainder is known \cite{Wang_2010}. Therefore, the computational complexities of the two algorithms primarily arise from the estimation of the common remainder $\hat{\mathsf r}^c$. For the fast MLE C-CRT, $\hat{\mathsf r}^c$ is obtained from the objective functions in (\ref{Re_rc_Im_rc}). Each evaluation of the real or imaginary part requires $4L$ real multiplications. According to Theorem \ref{Th_Omega_r_c}, obtaining $\hat{\mathsf r}^c$ necessitates $8L^2$ real multiplications in total. When $L\geq 5$, the common remainder search in the two-stage CRT arises from its second stage, where the objective functions are the special case of (\ref{Re_rc_Im_rc}) with $w_i=1$ for each $i$. According to the algorithm presented in \cite{Wang_2010}, each evaluation of the real or imaginary part involves $5(L-l)$ real multiplications, where $2l\leq L$ is the number of complex-valued moduli.
Denoting $\epsilon$ as the search step size, this algorithm requires at least $\frac{5ML}{\epsilon}$ real multiplications to estimate the common remainder. Since $\frac{M}{\epsilon}$ is generally much larger than $L$, the complexity of the two-stage CRT is higher in this case.
As mentioned earlier, the two-stage CRT proposed in \cite{Gong_2021} requires a search process to estimate the common remainder in \cite{Wang_2010} but does not utilize the fast algorithm proposed in \cite{Wang_2015}. If so, the number of searches would be $2(L-l)$ and only $10(L-l)^2$ real multiplications are required to obtain the common remainder. Consequently, the reference remainder can be properly determined and hence the folding integers ($n_i$ defined in \cite{Wang_2010}) can be correctly determined. However, the estimation is not the MLE since the estimation of the common remainder can not be utilized in the reconstruction of $\mathsf N$. Therefore, although the two-stage CRT provides a robust estimation, it is not optimal even when utilizing the fast algorithm proposed in \cite{Wang_2015}.

Fig. \ref{fzd_duibi} illustrates the numbers of the real multiplications required by the two-stage CRT in \cite{Gong_2021} and the proposed fast MLE C-CRT, where $\epsilon=0.001$. Since there is no need to search for the common remainder in the two-stage CRT when $L = 2, 3, 4$, the number of real multiplications is fewer than that of the fast MLE C-CRT. However, the two-stage CRT needs to search for the common remainder when $L\ge 5$, this leads to a significant increase in the number of real multiplications.

\begin{figure}[htbp]
\centering
\setlength{\belowcaptionskip}{-0.2cm}
\captionsetup{font={footnotesize}}
\captionsetup{labelsep=period}
\includegraphics[width=0.45\textwidth]{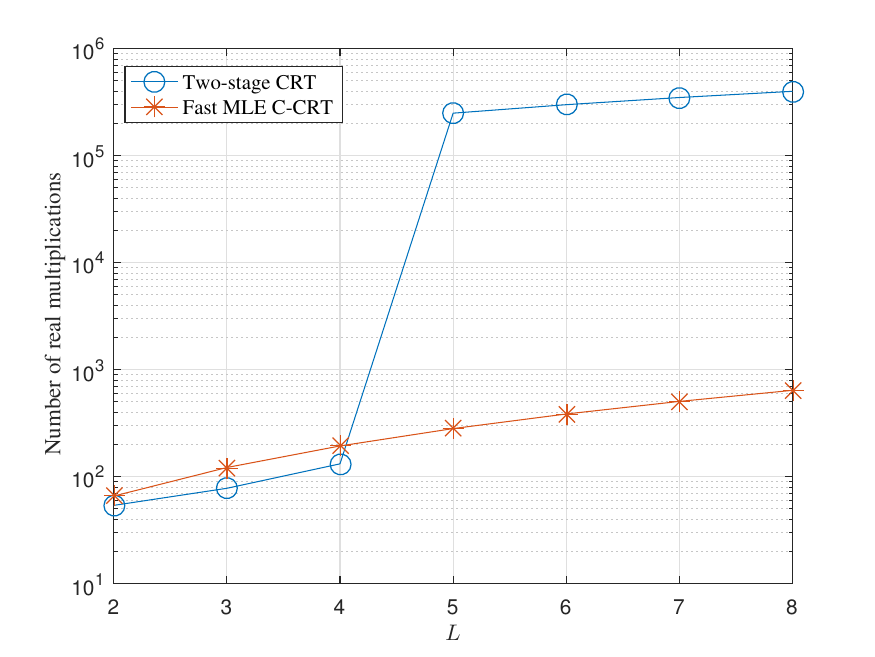}
\caption{Comparison of numbers of real multiplications.}\label{fzd_duibi}
\end{figure}

\subsection{Application in Modulo ADCs}
Now, we compare the three methods: the proposed fast MLE C-CRT, the two-stage CRT \cite{Gong_2021}, and two independent closed-form CRTs \cite{Wang_2010}.
For the proposed fast MLE C-CRT and the two-stage CRT, both the real and imaginary parts are sampled simultaneously from pairs of SR-ADCs with complex-valued moduli as illustrated in Fig. \ref{fig:C_SR_ADC}. For the closed-form CRT for real values, the real and imaginary parts are sampled separately from independent SR-ADCs with real-valued moduli using two sets of SR-ADCs as illustrated in Fig. \ref{fig:r_SR_ADC}.

The condition for the MLE C-CRT to uniquely reconstruct $\mathsf N$ from the congruence system (\ref{mathsfNmod}) is that $\mathsf N \in\mathcal F_{M\Gamma}$,
which is the same as that of the two-stage CRT in the sense of congruence.
Specifically, if the moduli in (\ref{mathsfNmod}) are the real numbers $M\Gamma_1, M\Gamma_2, \ldots, M\Gamma_L$, then we have
$$\mathrm{Re}(\mathsf{N}) \equiv \mathrm{Re}(\mathsf{r}_i) \mod M\Gamma_i, \ i = 1, 2, \ldots, L,$$
and
$$\mathrm{Im}(\mathsf{N}) \equiv \mathrm{Im}(\mathsf{r}_i) \mod M\Gamma_i, \ i = 1, 2, \ldots, L.$$
If $\mathrm{Re}(\mathsf{N})$ and $\mathrm{Im}(\mathsf{N})$ are within $[0, M\Gamma)$, then $\mathsf N$ can be uniquely reconstructed by the two independent real-valued CRTs.

For an integer $M\Gamma$ while $\mathsf\Gamma_1, \mathsf\Gamma_2, \ldots, \mathsf\Gamma_L$ are Gaussian integers, we know that $\mathcal F_{M\Gamma}$ is a square according to (\ref{FM}) as also mentioned earlier, where its sides are parallel to the real and imaginary axes with length $M\Gamma$. From Theorem \ref{Th_C_CRT}, any complex number in $\mathcal F_{M\Gamma}$ can be uniquely reconstructed by C-CRT, i.e., the multi-channel SR-ADCs in Fig. \ref{fig:C_SR_ADC}. Thus, both of the uniquely determinable real and imaginary parts of a complex number by using C-CRT or the multi-channel SR-ADCs are within $[0, M\Gamma)$. In this case, we call $M\Gamma$ as the dynamic range of the C-CRT. Thus, within this dynamic range, the sampled values of $\mathsf g(t)$ can be uniquely recovered by the C-CRT and the two-stage CRT with complex-valued moduli, or the closed-form CRTs with real-valued moduli. Note that to have an integer lcm, the complex-valued moduli can be selected as pairs of conjugate Gaussian integers.

When the maximum dynamic range of all the SR-ADCs is limited by $\Delta_{\max}$, the maximum dynamic range of the multi-channel SR-ADCs for real values is not greater than that for complex values and thus it is also limited by $\Delta_{\max}$. Since the real-valued moduli $\Gamma_i$ and the complex-valued moduli $\mathsf\Gamma_i$ belong to $\mathcal M_1=\left\{x\in \mathbb Z: 2\leq x\leq \Delta_{\max}\right\}$ and $\mathcal M_2=\left\{\mathsf x\in\mathbb Z
[\mathrm i]:\sqrt{2}\leq |\mathsf x|\leq \Delta_{\max} \right\}$,
respectively, and it is clear that $\mathcal M_1\subset\mathcal M_2$, there are more options for the complex-valued moduli, and its dynamic range is at least as large as that of the real-valued moduli. Table \ref{Table_DeltaMax} presents some examples for this application, where $M=1$, $L=3$, and the last column shows the dynamic ranges for both the real and imaginary parts of a uniquely determinable complex-valued signal using multi-channel SR-ADCs. For a fair comparison, when the maximal dynamic ranges of SR-ADCs are given, the sets of three pairwise coprime positive integers are optimized in Table \ref{Table_DeltaMax} in the closed-form CRTs for real values.

\begin{table}[h]
\centering
\captionsetup{font={footnotesize}}
\captionsetup{labelsep=period}
\caption{Comparison for real-valued and complex-valued moduli.}\label{Table_DeltaMax}
\begin{tabular}{cccc}
\toprule
$\Delta_{\max}$ & Method & $M\mathsf\Gamma_i$ & Dynamic range  $M\Gamma$ \\
\midrule
7 & Closed-form CRTs & 5, 6, 7 & $210$ \\
7 & C-CRT, Two-stage CRT & $4+5\mathrm{i}$, $4-5\mathrm{i}$, 7 & 287 \\
9 & Closed-form CRTs & 7, 8, 9 & $504$ \\
9 & C-CRT, Two-stage CRT & $7+4\mathrm{i}$, $7-4\mathrm{i}$, 9 & 585 \\
\bottomrule
\end{tabular}
\end{table}

In the simulations, we set a complex-valued bandlimited signal
$$ \mathsf g(t) = \sum_{k=-30}^{30}(a_k + \mathrm{i}b_k)A\cdot\mathrm{sinc}(t-k),$$
where $A$ is a constant, coefficients $a_k$ and $b_k$ are uniformly distributed in $[-1, 1]$, $\mathrm{sinc}(t) = \frac{\sin(\pi t)}{\pi t}$. For the closed-form CRTs, we set six SR-ADCs with dynamic ranges of $5, 5, 6, 6, 7$, and $7$ when $\Delta_{\max}=7$, and $7, 7, 8, 8, 9$, and $9$ when $\Delta_{\max}=9$. For the C-CRT and the two-stage CRT, we set six SR-ADCs with dynamic ranges of $\sqrt{41}$, $\sqrt{41}$, $\sqrt{41}$, $\sqrt{41}$, $7$ and $7$ when $\Delta_{\max}=7$, and $\sqrt{65}$, $\sqrt{65}$, $\sqrt{65}$, $\sqrt{65}$, $9$ and $9$ when $\Delta_{\max}=9$, where the moduli $M\mathsf\Gamma_i$ are shown in Table \ref{Table_DeltaMax}. The reconstruction error is quantified using the root relative squared error (RRSE), given by
$$  \mathsf g_{\mathrm{RRSE}} = \sqrt{\frac{\sum_{n}|\mathsf g(n/f_s) - \hat{\mathsf g}(n/f_s)|^2}{\sum_{n}|\mathsf g(n/f_s)|^2 }},$$
where the sampling frequency $f_s$ in the time domain for each channel is the Nyquist rate, i.e., $1$Hz.

Fig. \ref{fig_RRSE} illustrates the RRSE curves for the three methods in terms of SNR. The fast MLE C-CRT demonstrates the best performance overall. Similar to the previous RMSE simulations, the fast MLE C-CRT achieves robustness at an SNR of
$16$dB. In contrast, the two-stage CRT achieves robustness at an SNR of $18$dB. Additionally, for SNR values of $18$dB and higher, the performances of the two-stage CRT and the closed-form CRT are nearly indistinguishable.

\begin{figure}[h]
\centering
\setlength{\belowcaptionskip}{-0.5cm}
\captionsetup{font={footnotesize}}
\captionsetup{labelsep=period}
\subfigure[$\Delta_{\max}=7$ and $A=16$.]{\includegraphics[width=0.35\textwidth]{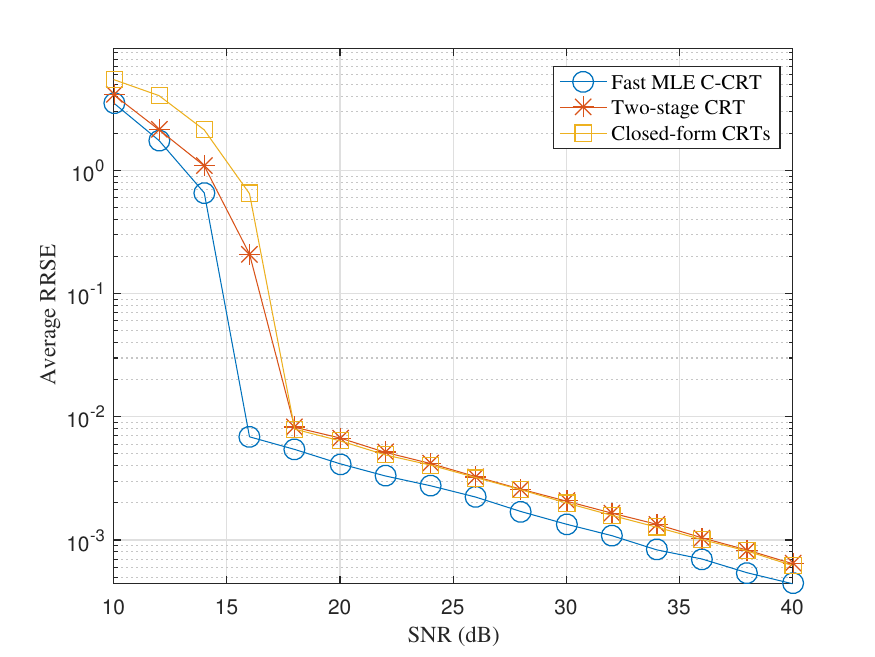}}
\subfigure[$\Delta_{\max}=9$ and $A=87$.]{\includegraphics[width=0.35\textwidth]{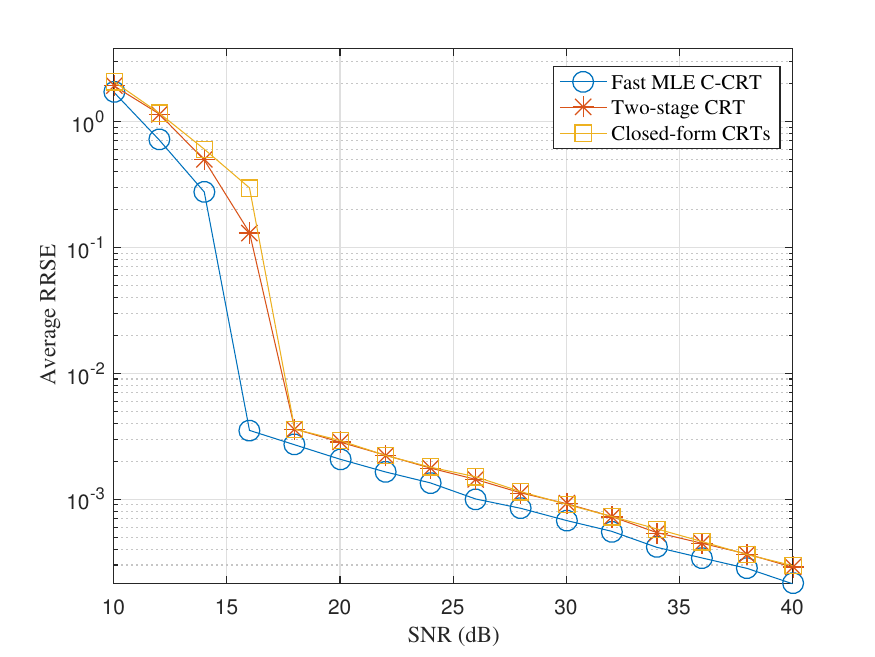}}
\caption{Average RRSE of the three methods.}\label{fig_RRSE}
\end{figure}

To show the high dynamic range of complex moduli, while the maximum dynamic range of all the SR-ADCs is constrained by $\Delta_{\max}$, we compute the TFR of the robust reconstruction for the three methods. The curves of the TFR for the three methods in terms of SNR are illustrated in Fig. \ref{fig_TFR}, where $a_k=-0.9$, $b_k=0.98$, and $\tau=0.25$. Since some sampled values exceed the dynamic range of SR-ADCs with real-valued moduli, the closed-form CRTs exhibit an error floor. Additionally, although the TFRs of both the fast MLE C-CRT and the two-stage CRT can approach $0$, the fast MLE C-CRT outperforms the two-stage CRT due to its optimal estimation of the common remainder. These results illustrate that the proposed MLE C-CRT for complex numbers performs better than the conventional CRT for real values. In addition, if C-CRT is thought of as a special 2D-CRT, it clearly shows that 2D-CRT (non-separable) performs better than two 1D-CRTs (separable).

\begin{figure}[h]
\centering
\setlength{\belowcaptionskip}{-0.2cm}
\captionsetup{font={footnotesize}}
\captionsetup{labelsep=period}
\subfigure[$\Delta_{\max}=7$ and $A=145$.]{\includegraphics[width=0.35\textwidth]{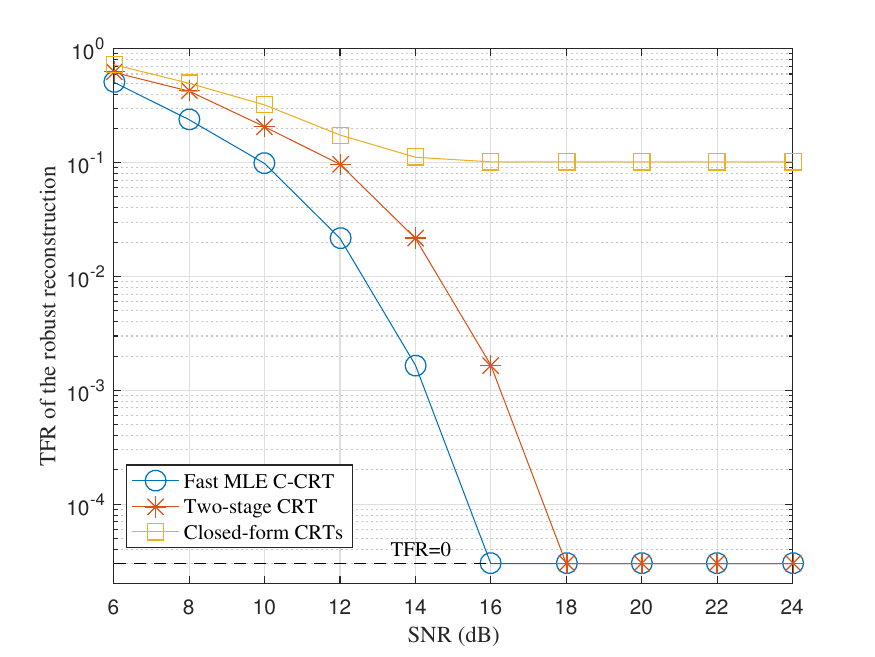}}
\subfigure[$\Delta_{\max}=9$ and $A=275$.]{\includegraphics[width=0.35\textwidth]{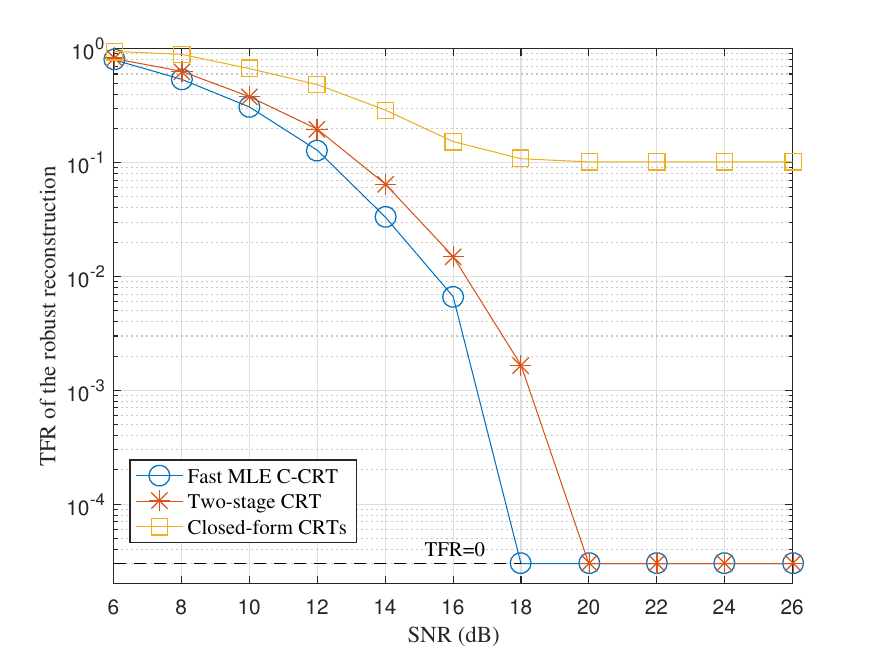}}
\caption{TFR of the three methods.}\label{fig_TFR}
\end{figure}

As a remark, for robust integer recovery from erroneous integer remainders using robust CRT, such as in \cite{Wang_2010}, the moduli are usually required to have a common integer factor $M>1$ and then the integer remainder errors can be as high as $\frac{M}{4}$ as a sufficient condition. This is for integers. As \cite{Wang_2015}, the robust reconstruction can be extended to real values, where although the integer moduli may not have a gcd $M>1$, such as those listed in Table \ref{Table_DeltaMax} where $M=1$, there still exists the robustness for real value reconstructions. This is because in the reconstruction of a real value, it has decimals. For example, when we consider one decimal precision in the reconstruction, it is equivalent to multiplying all the numbers including all the moduli by $10$, and then they become all integers. In this case the moduli have a gcd $10$ that provides the robustness for the robust CRT for integers.

\section{Conclusion}\label{sec6}
This paper proposes an efficient C-CRT algorithm based on MLE, enabling robust reconstruction of complex numbers from erroneous remainders modulo several Gaussian integers. The optimal common remainder can be determined using $L$ searches in both real and imaginary parts. Additionally, a necessary and sufficient condition for the C-CRT algorithm to achieve robust reconstruction is provided. In simulations, all the theoretical results are verified and it is illustrated that the proposed algorithm outperforms the current two-stage CRT and has been successfully applied to multi-channel SR-ADCs for complex-valued bandlimited signals.

\section*{Appendix}
\subsection{Proof of Properties \ref{prop1_F_M} to \ref{prop5_d_M}}
\textbf{Property \ref{prop1_F_M}}: Since $\mathsf z\in \mathcal F_{\mathsf M}$, we can express $\mathsf z$ as $\rho e^{\mathrm i\theta}(a+b\mathrm i)$ by (\ref{FM}),
where $0 \leq a, b < 1$. Thus, $\mathsf ze^{-\mathrm i\theta}=\rho(a+b\mathrm i)$.
According to (\ref{FM}), we have $\mathsf ze^{-\mathrm i\theta}\in\mathcal F_\rho$.

\textbf{Property \ref{prop2_F_M}}: For any point $\mathsf z$, multiplying it by
$e^{-\mathrm i\theta}$ results in a rotation of $\mathsf z$ by $-\theta$. By Property \ref{prop1_F_M}, we know that the region $\mathcal F_{\mathsf M}$ is rotated to $\mathcal F_{\mathsf\rho}$. Since $\mathcal F_{\mathsf\rho}$ is a square with side length $\rho$, we have that $\mathcal F_{\mathsf M}$ is also a square and the area of $\mathcal F_{\mathsf M}$ is $\rho^2$.

\textbf{Property \ref{prop1_d_M}}: Let $\mathsf x-\mathsf y = \mathsf M(a+b\mathrm i)$. By (\ref{d_mathsf_M}), we have
$\begin{aligned}
d_{\mathsf M}(\mathsf x, \mathsf y) = \mathsf M(a-[a]) +
\mathsf M (b - \left[b\right])\mathrm i.
\end{aligned}$ Since $-\frac{1}{2}\leq a-[a] < \frac{1}{2}$ and $-\frac{1}{2}\leq b-[b] < \frac{1}{2}$, we have $d_{\mathsf M}(\mathsf x, \mathsf y)\in\mathcal S_\mathsf M$.

\textbf{Property \ref{prop2_d_M}}:
Since $[\mathsf z+\mathsf k]=[\mathsf z]+\mathsf k$ holds for any complex number $\mathsf z$, we obtain by the definition of the circular distance in (\ref{d_mathsf_M}) that
$d_{\mathsf M}(\mathsf x +\mathsf {kM}, \mathsf y) = d_{\mathsf M}(\mathsf x, \mathsf y)$ and $d_\mathsf{\mathsf M}(\mathsf x, \mathsf y+\mathsf {kM})=d_\mathsf{\mathsf M}(\mathsf x, \mathsf y)$.
Furthermore, it follows from  (\ref{mathsfrN}) that
$$d_{\mathsf M}(\mathsf x, \mathsf y)=d_{\mathsf M}\left(\mathsf x, \mathsf y-\mathsf M\left\lfloor \frac{\mathsf y}{\mathsf M} \right\rfloor \right) = d_{\mathsf M}(\mathsf x , \left\langle \mathsf y\right\rangle_\mathsf M).$$

\textbf{Property \ref{prop3_d_M}}: Let $\mathsf x-\mathsf y=\mathsf M(c+d\mathrm i)$, where $-\frac{1}{2}\leq c, d < \frac{1}{2}$. Note that $\left[\frac{\mathsf M(c+d\mathrm i)}{\mathsf M}\right]=0$. Hence, $d_\mathsf{M}(\mathsf x, \mathsf y)= \mathsf x-\mathsf y$.

\textbf{Property \ref{prop4_d_M}}: Let $\mathsf x-\mathsf y = \mathsf M(a+b\mathrm i)$. Then, we have either $\mathsf x-\mathsf y \in \mathcal S_\mathsf M$ or $\mathsf x-\mathsf y\not\in \mathcal S_\mathsf M$. For the first case, we have three subcases: 1) $-\frac{1}{2} < a < \frac{1}{2}, b = -\frac{1}{2}$; 2) $-\frac{1}{2} < b < \frac{1}{2}$, $a = -\frac{1}{2}$; 3) $a=b=-\frac{1}{2}$. By Property \ref{prop3_d_M}, we have $|d_{\mathsf M}(\mathsf x, \mathsf y)|=|\mathsf x-\mathsf y|$ for these three subcases. For the second case,  we have three subcases: 1) $-\frac{1}{2}\leq a < \frac{1}{2}$ and $b=\frac{1}{2}$; 2) $-\frac{1}{2}\leq b < \frac{1}{2}$ and $a=\frac{1}{2}$; 3) $a=b=\frac{1}{2}$. Without loss of generality, we only prove subcase 1). Since $\left[\frac{\mathsf M(a+b\mathrm i)}{\mathsf M}\right]=\mathrm i$, we have $|d_\mathsf M(\mathsf x,\mathsf y)|=|\mathsf M(a+b\mathrm i) - \mathsf M\mathrm i| = |\mathsf M(a-\frac{1}{2}\mathrm i)| = |\mathsf M(a+\frac{1}{2}\mathrm i)| = |\mathsf x-\mathsf y|$.

\textbf{Property \ref{prop5_d_M}}: It suffices to prove $d_k(\mathsf x, \mathsf y)\in\mathcal S_{k\mathsf M}\cup\partial\mathcal S_{k\mathsf M}$ by Properties \ref{prop3_d_M} and \ref{prop4_d_M}.
Since $\mathcal S_{k\mathsf M}\cup\partial\mathcal S_{k\mathsf M}$ is a square with side length $|k\mathsf M|$,
its incircle is $\mathcal O=\left\{\mathsf z: |\mathsf z|\leq \frac{|k\mathsf M|}{2}\right\}$.
An inscribed square of
$\mathcal O$ is $\mathcal Q = \left\{a+b\mathrm i: -\frac{\sqrt{2}}{2}|k\mathsf M|\leq a, b\leq \frac{\sqrt{2}}{2}|k\mathsf M| \right\}$.
Since $|\mathsf M|\geq \sqrt{2}$, we have
$\mathcal S_k=\left\{a+b\mathrm i: -|k|\leq a, b< |k| \right\}\subseteq\mathcal Q\subseteq\mathcal S_{k\mathsf M}\cup\partial\mathcal S_{k\mathsf M}$.
Thus, $d_k(\mathsf x, \mathsf y)\in\mathcal S_{k\mathsf M}\cup\partial\mathcal S_{k\mathsf M}$.

\subsection{Proof of Theorem \ref{thm_Delta_r}}\label{appendix_1}
\begin{proof}
By (\ref{robust_ri_SM}), we have
$$-\frac M 2\leq\sum_{\Delta \mathsf r_i\in \mathcal V}\frac{w_i\mathrm{Re}(\Delta \mathsf r_i)}{\sum_{\Delta\mathsf r_j\in \mathcal V}w_j} - \sum_{\Delta \mathsf r_i\in\overline{\mathcal V}}\frac{w_i\mathrm{Re}(\Delta \mathsf r_i)}{\sum_{\Delta \mathsf r_j\in\overline{\mathcal V}}w_j} < \frac M 2.$$
According to Theorem 3 in \cite{Wang_2015}, we have
$$\mathrm{Re}(\hat{\mathsf r}^c) = \langle \mathrm{Re}(\mathsf r^c) + \mathrm{Re}(\overline{\Delta \mathsf r})\rangle_M.$$
Hence, (\ref{Re drc}) holds. Similarly, we have
$$\mathrm{Im}(\hat{\mathsf r}^c) = \langle \mathrm{Im}(\mathsf r^c) + \mathrm{Im}(\overline{\Delta \mathsf r})\rangle_M.$$
This leads to (\ref{Im drc}). Since $M\in\mathbb Z$, we have
$$\hat{\mathsf r}^c
=\langle \mathrm{Re}(\mathsf r^c)+(\mathrm{Re}\overline{\Delta \mathsf r})\rangle_M + \mathrm i\langle \mathrm{Im}(\mathsf r^c) + \mathrm{Im}(\overline{\Delta \mathsf r})\rangle_M
=\langle \mathsf r^c + \overline{\Delta \mathsf r}\rangle_M.$$
\end{proof}

\subsection{Proof of Theorem \ref{Theo_iff}}\label{appendix_2}
\begin{proof}
The sufficiency has been proven, now we prove the necessity. 
Since $\hat{\mathsf N} - \mathsf N = \overline{\Delta \mathsf r}$, we have 
$$M\hat{\mathsf N}_0 + \hat{\mathsf r}^c - M\mathsf N_0 - \mathsf r^c = \overline{\Delta \mathsf r}.$$
This leads to
$$\hat{\mathsf r}^c = \langle \mathsf r^c + \overline{\Delta \mathsf r} \rangle_M.$$
Hence, 
$$\mathrm{Re}(\hat{\mathsf r}^c) = \langle \mathrm{Re}(\mathsf r^c) + \mathrm{Re}(\overline{\Delta \mathsf r})\rangle_M, \ \mathrm{Im}(\hat{\mathsf r}^c) = \langle \mathrm{Im}(\mathsf r^c) + \mathrm{Im}(\overline{\Delta \mathsf r})\rangle_M.$$

If there exists a non-empty set $\mathcal V_1 \subset \mathcal U$ satisfying
$$\sum_{\Delta \mathsf r_i\in \mathcal V_1}\frac{w_i\mathrm{Re}(\Delta\mathsf r_i)}{\sum_{\Delta\mathsf r_j \in \mathcal V_1}w_j} - \sum_{\Delta\mathsf r_i\in\overline{\mathcal V}_1}\frac{w_i\mathrm{Re}(\Delta\mathsf r_i)}{\sum_{\Delta\mathsf r_j\in\overline{\mathcal V}_1}w_j} \geq \frac{M}{2},$$
then we can obtain
$$\sum\limits_{\Delta \mathsf r_i\in \mathcal V_1}w_i\left(\mathrm{Re}(\Delta\mathsf r_i - \overline{\Delta\mathsf r})\right) \geq \frac{M}{2}\left(1-\sum\limits_{\Delta\mathsf r_i\in \mathcal V_1}w_i\right)\sum\limits_{\Delta\mathsf r_i\in \mathcal V_1}w_i.$$
Let 
\begin{equation}\label{eq_tilde_rc}
\tilde{\mathsf r}^c = \left\langle \mathsf r^c + \overline{\Delta \mathsf r} - M\sum_{\Delta \mathsf r_i\in \mathcal V_1}w_i\right\rangle_M.
\end{equation}
Then we have
\begin{equation}\label{Ineq_eq}
\sum_{i=1}^{L}w_id_M^2\left(\mathrm{Re}(\tilde{\mathsf r}_i^c), \langle \mathrm{Re}(\mathsf r^c) + \mathrm{Re}(\overline{\Delta\mathsf r})\rangle_M\right)
 \geq \sum_{i=1}^{L}w_id_M^2 \left(\mathrm{Re}(\tilde{\mathsf r}_i^c),
\mathrm{Re}(\tilde{\mathsf r}^c)\right).
\end{equation}
If the equality of \eqref{Ineq_eq} holds, then $\mathrm{Re}(\tilde{\mathsf r}^c)$ is optimal. Hence,
$$\mathrm{Re}(\tilde{\mathsf r}^c) = \mathrm{Re}(\hat{\mathsf r}^c) = \langle \mathrm{Re}(\mathsf r^c) + \mathrm{Re}(\overline{\Delta \mathsf r})\rangle_M.$$
By \eqref{eq_tilde_rc}, we can obtain that either $\mathcal V_1 = \emptyset$ or $\mathcal V_1  = \mathcal U$ holds, which is a contradiction.
If the inequality of \eqref{Ineq_eq} holds, then
$\mathrm{Re}(\hat{\mathsf r}^c)$ is not optimal, which is a contradiction.
\end{proof}


\begin{thebibliography}{99}
\bibitem{Krishna_1994}
H. Krishna, B. Krishna, K.-Y. Lin, and J.-D. Sun, {\em Computational Number Theory and Digital Signal Processing: Fast Algorithms and Error Control Techniques,} Boca Raton, FL: CRC, 1994.

\bibitem{Ding_1999}
C. Ding, D. Pei, and A. Salomaa, {\em Chinese Remainder Theorem: Applications in Computing, Coding, Cryptography,} Singapore: World Scientific, 1999.

\bibitem{Xia_2007}
X.-G. Xia and G. Y. Wang, ``Phase unwrapping and a robust Chinese remainder theorem,''  {\em IEEE Signal Process. Lett.,} vol. 14, no. 4, pp. 247-250, Apr. 2007.


\bibitem{Li_2007}
G. Li, J. Xu, Y.-N. Peng, and X.-G. Xia, ``An efficient implementation of a robust phase-unwrapping algorithm,'' {\em IEEE Signal Process. Lett.,} vol. 14, no. 6, pp. 393-396, Jun. 2007.

\bibitem{Li_2008}
X. W. Li and X.-G. Xia, ``A fast robust Chinese remainder theorem based phase unwrapping algorithm,'' {\em IEEE Signal Process. Lett.,} vol. 15, no. 10, pp. 665-668, Oct. 2008.

\bibitem{Li_2009}
X. W. Li, H. Liang, and X.-G. Xia, ``A robust Chinese remainder theorem with its applications in frequency estimation from undersampled waveforms,'' {\em IEEE Trans. Signal Process.,} vol. 57, no. 11, pp. 4314-4322, Nov. 2009.

\bibitem{Wang_2010}
W. J. Wang and X.-G. Xia, ``A closed-form robust Chinese remainder theorem and its performance analysis,'' {\em IEEE Trans. Signal Process.,} vol. 58, no. 11, pp. 5655-5666, Nov. 2010.

\bibitem{Yang_2014}
B. Yang, W. J. Wang, X.-G. Xia, and Q. Y. Yin, ``Phase detection based range estimation with a dual-band robust Chinese remainder theorem,'' {\em Sci. China-Inf. Sci.,} vol. 57, no. 2, pp. 1-9, Feb. 2014.

\bibitem{Xiao_2014}
L. Xiao, X.-G. Xia, and W. J. Wang, ``Multi-stage robust Chinese remainder theorem,'' {\em IEEE Trans. Signal Process.,} vol. 62, no. 18, pp. 4772-4785, Sep. 2014.

\bibitem{Wang_2015}
W. J. Wang, X. P. Li, W. Wang, and X.-G. Xia, ``Maximum likelihood estimation based robust chinese remainder theorem for real numbers and its fast algorithm,'' {\em IEEE Trans. Signal Process.,} vol. 63, no. 13, pp. 3317-3331, Jul. 2015.


\bibitem{Ore_1952}
O. Ore, ``The general Chinese remainder theorem,'' {\em Amer. Math. Month.,} vol. 59, no. 6, pp. 365-370, Jun. 1952.


\bibitem{Xu_2018}
J. Xu, Z.-Z. Huang, Z.-R. Wang, L. Xiao, X.-G. Xia, and T. Long, ``Radial velocity retrieval for multichannel SAR moving targets with time-space Doppler deambiguity,'' {\em IEEE Trans. Geosci. Remote Sens.,} vol. 56, no. 1, pp. 35-48, Jan. 2018.

\bibitem{Huang_2024}
X. R. Huang, Z. H. Yuan, and L. F. Chen, ``A robust multi-channel InSAR phase unwrapping method based on grouped Chinese remainder theorem,'' in {\em Proc. IEEE Int. Conf. Electro Inf. Technol.,} Chengdu, China, 2024, pp. 701-705.

%\bibitem{Guessoum_1986}
%A. Guessoum and R. Mersereau, ``Fast algorithms for the multidimensional discrete Fourier transform,'' {\em IEEE Trans. Acoustics, Speech, Signal Process.,} vol. 34, no. 4, pp. 937-943, Aug. 1986.

\bibitem{Xiao_2020}
L. Xiao, X.-G. Xia, and Y.-P. Wang, ``Exact and robust reconstructions of integer vectors based on multidimensional Chinese remainder theorem (MD-CRT),'' {\em IEEE Trans. Signal Process.,} vol. 68, pp. 5349-5364, Sep. 2020.

\bibitem{Xiao_2024}
L. Xiao, H. Y. Huo, and X.-G. Xia, ``Robust multidimentional Chinese remainder theorem for integer vector reconstruction,'' {\em IEEE Trans. Signal Process.,} vol. 72, pp. 2364-2376, May 2024.
%%%
\bibitem{Xia_1997}
X.-G. Xia and G. C. Zhou, ``Multiple frequency detection in undersampled waveforms,'' in {\em  Proc. 31st Annu. Asilomar Conf. Signals, Systems, and Compuers,} Pacific Grove, California, 1997, pp. 867-871.

\bibitem{Xia_1999}
X.-G. Xia, ``On estimation of multiple frequencies in undersampled complex valued waveforms,'' {\em IEEE Trans. Signal Process.,} vol. 47, no. 12, pp. 3417-3419, Dec. 1999.

\bibitem{Wang_GCRT_2015}
W. Wang, X. P. Li, X.-G. Xia, and W. J. Wang, ``The largest dynamic range of a generalized Chinese remainder theorem for two integers,'' {\em IEEE Signal Process. Lett.,} vol. 22, no. 2, pp. 254-258, Feb. 2015.

\bibitem{Li_2016}
X. P. Li, X.-G. Xia, W. J. Wang, and W. Wang, ``A robust generalized Chinese remainder theorem for two integers,'' {\em IEEE Trans. Inf. Theory,} vol. 62, no. 12, pp. 7491-7504, Dec. 2016.

\bibitem{Li_2019}
X. P. Li, T.-Z. Huang, Q. Y. Liao, and X.-G. Xia, ``Optimal estimates of two common remainders for a robust generalized Chinese remainder theorem,'' {\em IEEE Trans. Signal Process.,} vol. 67, no. 7, pp. 1824-1837, Apr. 2019.

\bibitem{Xiao_2015_p}
L. Xiao and X.-G. Xia, ``Error correction in polynomial remainder codes with non-pariwise coprime moduli and robust Chinese remainder theorem for polynomials,'' {\em IEEE Trans. Commun.,} vol. 63, no. 3, pp. 605-616, Mar. 2015.

\bibitem{Xiao_2018_p}
L. Xiao and X.-G. Xia, ``Robust polynomial reconstruction via Chinese remainder theorem in the presence of small degree residue errors,'' {\em IEEE Trans. Circuits Syst. II, Exp. Briefs,} vol. 65, no. 11, pp. 1778-1782, Nov. 2018.

\bibitem{Gong_2021}
Y. C. Gong, L. Gan, and H. Q. Liu, ``Multi-channel modulo samplers constructed from Gaussian integers,'' {\em IEEE Signal Process. Lett.,} vol. 28, no. 8, pp. 1828-1832, Aug. 2021.

\bibitem{Lin_1994}
Y.-P. Lin, S.-M. Phoong, and P. P. Vaidyanathan, ``New results on multidimensional Chinese remainder theorem,'' {\em IEEE Signal Process. Lett.,} vol. 1, no. 11, pp. 176-178, Nov. 1994.

\bibitem{Gan_2020}
L. Gan and H. Q. Liu, ``High dynamic range sensing using multi-channel modulo samplers,'' in {\em Proc. 11th IEEE Sensor Array Multichannel Signal Process. Workshop,} Hangzhou, China, 2020, pp. 1-5.

\bibitem{Yan_2024}
W. Y. Yan, L. Gan, S. Q. Hu, and H. Q. Liu, ``Towards optimized multi-channel modulo-ADCs: Moduli selection strategies and bit depth analysis,'' in {\em Proc. IEEE Int. Conf. Acoust., Speech, Sig. Proc.,} Seoul, Korea, 2024, pp. 9496-9500.

\bibitem{Yan_2025}
W. Y. Yan, L. Gan, and Y. D. Zhang, ``Threshold sensitivity in two-channel modulo ADCs: analysis and robust reconstruction," in {\em Proc. IEEE Int. Conf. Acoust., Speech, Sig. Proc.,} Hyderabad, India, 2025, pp. 1-5.

%\bibitem{SRADC}
%V. Pavl\'{i}\v{c}ek, R. Guo, and A. Bhandari, ``Bits, channels, frequencies and unlimited sensing: Pushing the limits of sub-Nyquist prony,'' in {\em Proc.\ EUSIPCO}, 2024, pp. 2462-2466.

\bibitem{Vaidyanathan_1993}
P. P. Vaidyanathan, {\em Multirate Systems and Filter Banks}, Englewood Cliffs, NJ, USA: Prentice-Hall, 1993.

\bibitem{Dummit_2004}
D. S. Dummit and R. M. Foote, {\em Abstract Algebra,} 3nd Ed., Hoboken, NJ, USA: Wiley, 2004.

\bibitem{Li_2019_def}
C. H. Li, L. Gan, and C. Ling, ``Coprime sensing via Chinese remaindering over quadratic fields-Part I: Array designs,'' {\em IEEE Trans. Signal Process.,} vol. 67, no. 11, pp. 2898-2910, Jun. 2019.

\bibitem{LHLi_2019}
C. H. Li, L. Gan, and C. Ling, ``Coprime sensing via Chinese remaindering over
quadratic fields-Part II: Generalizations and applications,'' {\em IEEE Trans. Signal Process.,} vol. 67, no. 11, pp. 2911-2922, Jun. 2019.

\bibitem{PP_2011}
P. Pal and P. P. Vaidyanathan, ``Coprimality of certain families of integer matrices,'' {\em IEEE Trans. Signal Process.,} vol. 59, no. 4, pp. 1481-1490, Apr. 2011.

\bibitem{Greco_2021}
L. Greco, G. Saraceno, and C. Agostinelli, ``Robust fitting of a wrapped normal model to multivariate circular data and outlier detection,'' {\em Stats,} vol. 4, no. 2, pp. 454-471, Jun. 2021.
\end{thebibliography}
\end{document}